\documentclass[journal]{IEEEtran}
\usepackage{hyperref}
\usepackage{graphicx}
\usepackage{cite}
\usepackage{epstopdf}
\usepackage{epsfig}
\usepackage{amsthm}
\usepackage{amsmath}
\usepackage{amssymb}
\usepackage{mathptmx}
\ifCLASSOPTIONcompsoc
	\usepackage[caption=false,font=normalsize,labelfont=sf,textfont=sf]{subfig}
\else
	\usepackage[caption=false,font=footnotesize]{subfig}
\fi

\ifCLASSOPTIONcaptionsoff
	\usepackage[nomarkers]{endfloat}
	\let\MYoriglatexcaption\caption
	\renewcommand{\caption}[2][\relax]{\MYoriglatexcaption[#2]{#2}}
\fi

\newtheorem{theorem}{Theorem}

\newtheorem{remark}{Remark}
\newtheorem{conjecture}{Conjecture}

\begin{document}

\title{Distributed area coverage control with imprecise robot localization: Simulation and experimental studies}

\author{Sotiris~Papatheodorou,
        Anthony~Tzes \emph{Senior Member, IEEE},
        Konstantinos~Giannousakis
        and~Yiannis~Stergiopoulos
\thanks{This work has received funding from the European Union's Horizon 2020 Research and Innovation Programme under the Grant Agreement No.644128, AEROWORKS.}
\thanks{The authors are with the Electrical \& Computer Engineering Department, University of Patras, Rio, Achaia 26500, Greece. Corresponding author's email: \small\tt{tzes@ece.upatras.gr}}}
%
\markboth{IEEE Transactions on Automatic Control}%
{}
%
%
\maketitle
\begin{abstract}
This article examines the area coverage problem for a network of mobile robots with imprecise agents’ localization. Each robot has uniform radial sensing ability, governed by first order kinodynamics. The convex-space is partitioned based on the Guaranteed Voronoi (GV) principle and each robot's area of responsibility corresponds to its GV-cell, bounded by hyperbolic arcs. The proposed control law is distributed, demanding the positioning information about its GV-Delaunay
neighbors. Simulation and experimental studies are offered to highlight the efficiency of the proposed control law.
\end{abstract}


\IEEEpeerreviewmaketitle
\section{Introduction}
\IEEEPARstart{A}{rea} coverage of a planar region by a swarm of robotic agents is an active field of research with great interest and potential. The task involves the dispersal inside a region of interest of autonomous, sensor--equipped robotic agents in order to achieve sensor coverage of that region. Distributed control algorithms are a popular choice for the task because of their computational efficiency compared to centralized controllers, their ability to adapt in real--time to changes in the robot swarm and their reliance only on local information. Thus they can be implemented on robots with low computational power and low power radio transceivers, since each agent is required to communicate only with its neighbors. The downside of distributed control schemes is that they lead to local optima of their objective function because of their lack of information on the state of all agents. 

Area coverage can be either static \cite{Mahboubi_IEEETAC2016,vanEeden_SICE2016,Simonetto_Allerton2012} in which case the agents converge to some static final configuration or sweep \cite{Song_Automatica2013,Franco_EJC2015a,Zhai_Automatica2013} in which the agents constantly move in order to satisfy a constantly changing coverage objective. 

Other factors that need to be taken into account are the type of the region surveyed \cite{Stergiopoulos_IEEETAC15,Alitappeh_SC2016}, the sensing pattern and performance of the sensors \cite{Stergiopoulos_IETCTA10,Arslan_ICRA2016,Pimenta_CDC08,Stergiopoulos_ICRA14}, or the particular dynamics of the agents \cite{Luna_CDC2010}. 

There have been varied approaches to area coverage such as distributed optimization \cite{Cortes_SIAMJCO05,Cortes_ESAIMCOCV05}, model predictive control \cite{Nguyen_MED2016,Mohseni_IEEESJ2016}, game theory \cite{Ramaswamy_ACC2016} and optimal control \cite{Nguyen_ISIC2016}.

All localization systems have an inherent uncertainty in their measurements, thus there has been significant work on taking this uncertainty into account either by data fusion \cite{Hage_MED2016}, safe trajectory planning \cite{Davis_IEEERAL2016} or spatial probability \cite{Habibi_IEEETAC2016}. Our approach creates a suitable space partioning scheme in lieu of the agents' positioning uncertainty and a distributed gradient--based control law.

This article examines the problem of persistent area coverage of a convex planar region by a network of homogeneous agents equipped with isotropic sensors.  In order to take into account the localization uncertainty, each agent is assumed to lie somewhere inside a circular positioning uncertainty region and Guaranteed Voronoi partitioning \cite{Evans_CCCG2008} of the region of interest is constructed from these uncertain regions. Subsequently, a distributed control law is derived based on the aforementioned partitioning so that a coverage objective increases monotonously. Because of the complexity of that law, a suboptimal one is proposed and both laws are compared through simulations. Additionally, the suboptimal control law is evaluated by two experiments using differential drive robots. This article is an extension of \cite{Papatheodorou_MED2016} which did not contain simulation results from the optimal control law and lacked the experimental evaluation of the suboptimal control law.

The layout of the article is as follows. The required mathematical preliminaries are presented in Section \ref{sec:preliminaries}. The Voronoi and the Guaranteed Voronoi space partitions are examined in Section \ref{sec:partitioning}. Section \ref{sec:problem_law} contains the definition of the coverage objective and the derivation of the gradient ascent based control law that maximizes it. Additionally it contains the proposed suboptimal control law. Simulation studies are presented in Section \ref{sec:simulations} in order to compare and evaluate both control laws. Section \ref{sec:experiments} contains experimental results from the implementation of the suboptimal control law and is followed by concluding remarks.
\section{Mathematical Preliminaries}
\label{sec:preliminaries}
Let us assume a compact convex region $\Omega \in \mathbb{R}^2$ under surveillance and a network of $n$ identical dimensionless mobile agents (nodes). Each agent's dynamics are governed by
\begin{equation}
\dot{q}_i = u_i, ~u_i \in \mathbb{R}^2, ~q_i \in \Omega, ~i \in I_n
\label{eq:dynamics}
\end{equation}
where $u_i$ is the control input for each agent and $I_n = \{1,~2,\dots, n\}$. 

The agents' positions are not known precisely and thus in the general case each agent has a positioning uncertainty of different measure. We assume that an upper bound $r_i^u$ for the positioning uncertainty of each agent is known and thus its center lies within a disk
\begin{equation}
C_i^u (q_i,r_i^u) = \left\{ q \in \Omega \colon \parallel q-q_i \parallel \leq r_i^u \right\}, ~i \in I_n
\label{eq:positioning_uncertainty}
\end{equation}
where $\parallel \cdot \parallel$ is the Euclidean metric and $q_i$ the agent's position as reported by its localization equipment.

All agents are equipped with identical omnidirectional sensors with circular sensing footprints and as such the sensed region of each agent can be defined as
\begin{equation*}
C_i^s (q_i,r^s) = \left\{ q \in \Omega \colon \parallel q-q_i \parallel \leq r^s \right\}, ~i \in I_n
\end{equation*}
where $r^s$ is the common range of the sensors.

We also define for each agent a `Guaranteed Sensed Region' (GSR) as the region that is sensed by the agent for all of its possible positions inside $C_i^u$. The GSRs are defined as
\begin{equation*}
C_i^{gs}(C_i^u,C_i^s) = \left\{ \bigcap_{q_i} C_i^s(q_i,r^s), ~\forall q_i \in C_i^u \right\}, ~i \in I_n
\end{equation*}
and since $C_i^u$ and $C_i^s$ are disks, the previous expression can be further simplified into
\begin{equation}
C_i^{gs}(q_i,r^s-r_i^u) = \left\{ q \in \Omega: \parallel q-q_i\parallel \leq r^s-r_i^u \right\},~i \in I_{n},
\label{eq:guaranteed_sensing}
\end{equation}
provided that $r^s \geq r_i^u$. 

When an agent's position is known precisely, i.e. $r_i^u = 0$, its GSR is equivalent to its sensed region, whereas when an agent's positioning uncertainty is too large compared to its sensing capabilities, i.e. $r_i^u > r^s$, then $C_i^{gs} = \emptyset$.


\section{Space Partitioning}
\label{sec:partitioning}
Most distributed area coverage schemes assign an area of responsibility to each agent based on information from its neighboring agents. Then each agent is responsible for maximizing the coverage solely on its own responsibility region. The most common partitioning scheme is the Voronoi partitioning, described in Section \ref{sec:voronoi}. However in this article, because of the localization uncertainty of the agents, the guaranteed Voronoi diagram described in Section \ref{sec:gvoronoi} is used for the assignment of areas of responsibility.

\subsection{Voronoi Diagram}
\label{sec:voronoi}

The Voronoi diagram for a set of points $Q = \{q_1, q_2 \dots q_n \}, ~q_i \in \Omega$ is defined as follows
\begin{equation*}
V_i = \left\{ q \in \Omega \colon \left\| q-q_i \right\|
\leq\left\|q-q_j \right\|,~~\forall
j\in I_n,~j\neq i\right\},~~i\in I_n.
\end{equation*}
Each $V_i \subset \Omega$ is called the Voronoi cell of node $i$. A Voronoi diagram of 6 points constrained in a subset $\Omega$ of $\mathbb{R}^2$ is shown in Figure \ref{fig:V_GV_comparison} (left), where the cell boundaries are shown in blue.

It is important to note that the Voronoi diagram is a complete tessellation of $\Omega$, that is $\bigcup_{i \in I_n} V_i = \Omega$. 

Another useful property of Voronoi diagrams is their duality to Delaunay triangulations. Let us define the Delaunay neighbors of a point $q_i$ as
\begin{equation}
	N_i = \left\{j\in I_n,~j\neq i \colon V_i\cap V_j \neq \emptyset\right\},~i\in I_n,
	\label{eq:delaunay}
\end{equation}
When constructing the Voronoi cell of a point $q_i$, only its Delaunay neighbors need to be considered, thus simplifying the definition of its Voronoi cell into
\begin{equation*}
V_i = \left\{ q \in \Omega \colon \left\| q-q_i \right\|
\leq\left\|q-q_j \right\|,~~\forall
j\in N_i\right\},~~i\in I_n.
\end{equation*}
Using this property, a Voronoi diagram can be constructed by first finding the Delaunay triangulation of the input points and then for each cell $V_i$ by finding the intersection of the halfplanes defined by point $q_i$ and all points in $N_i$.

\begin{figure}[htb]
	\centering
	\ifx\singlecol\undefined
		\includegraphics[width=0.23\textwidth]{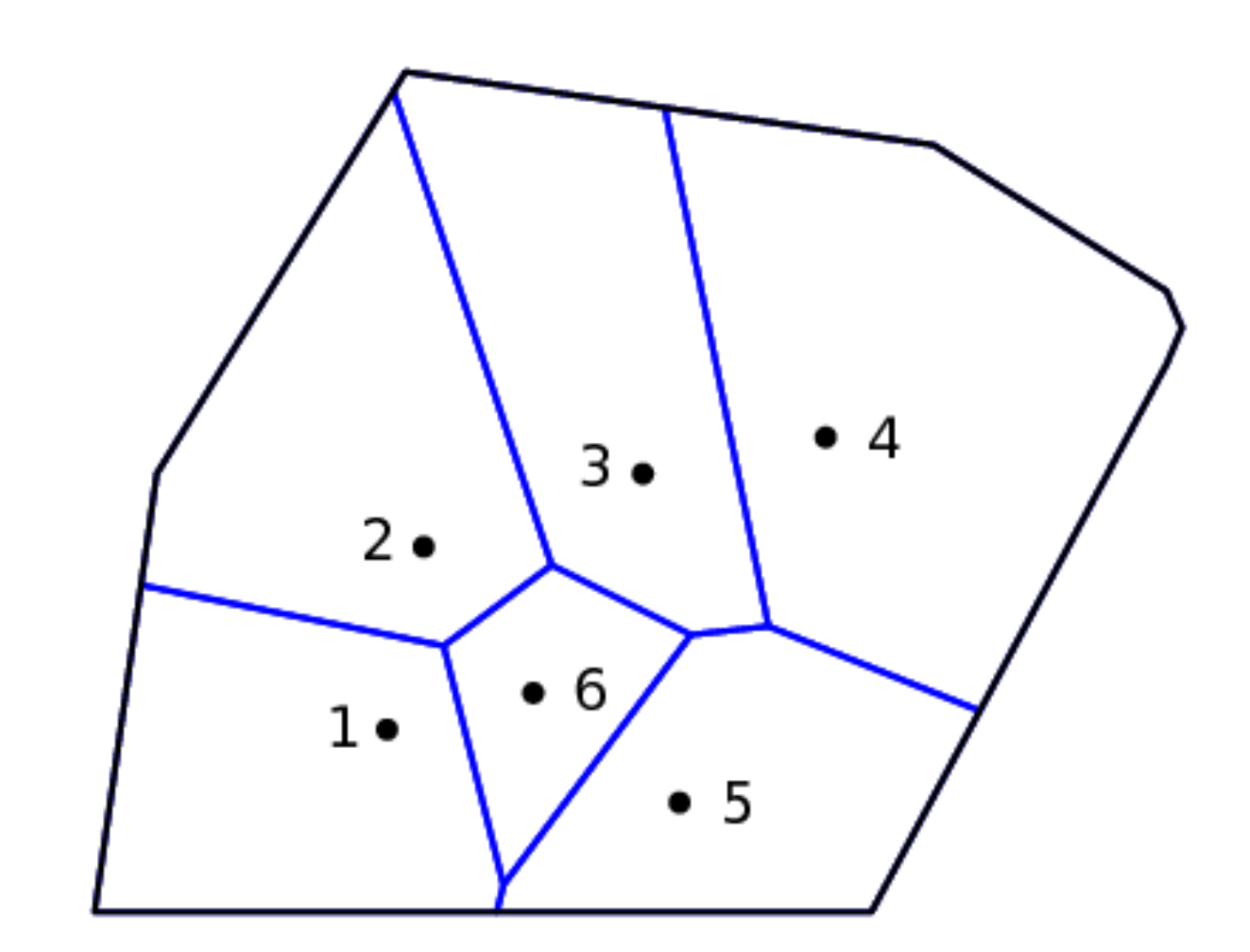}\hspace{0.01cm}
		\includegraphics[width=0.23\textwidth]{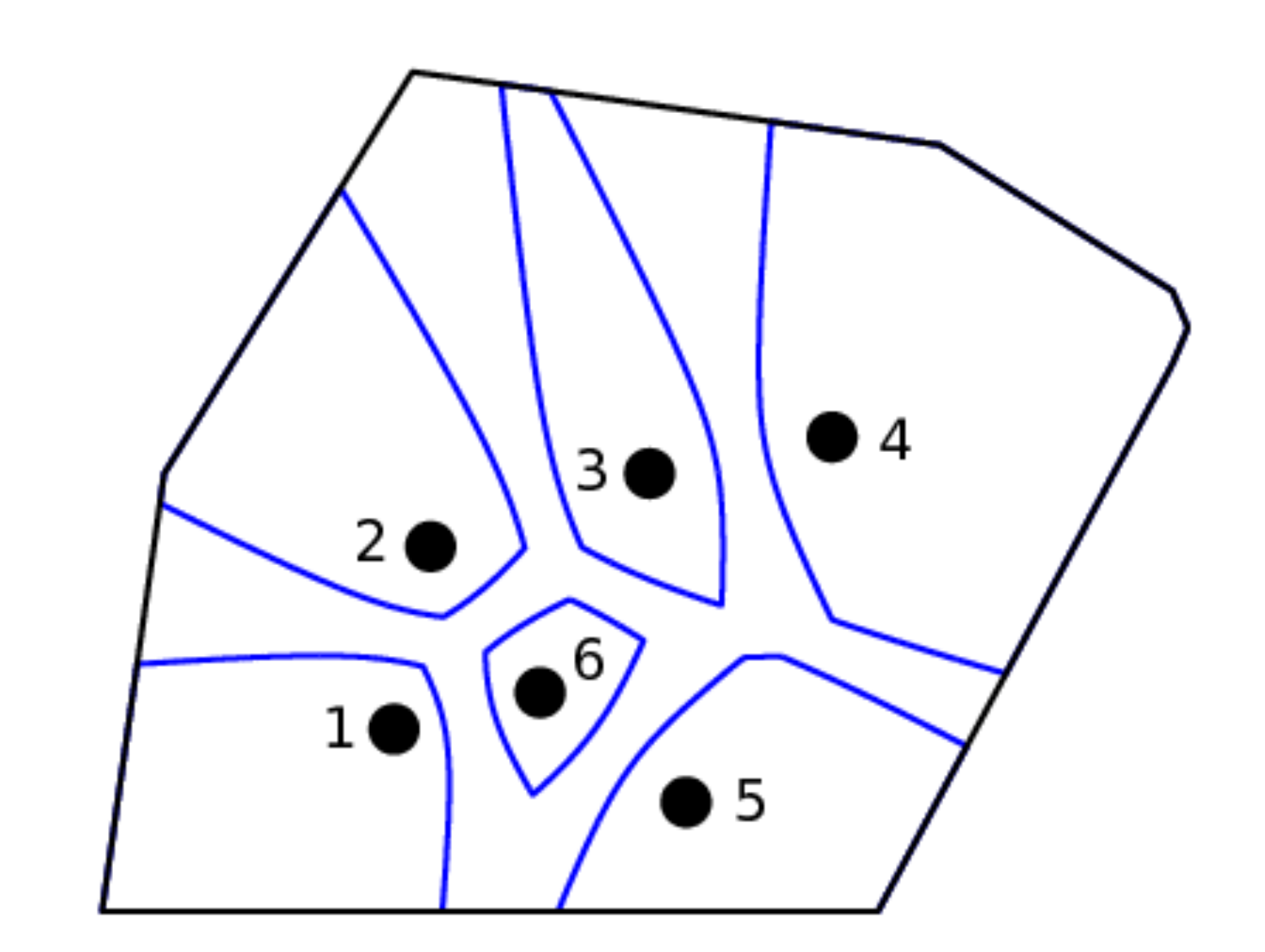}
	\else
		\includegraphics[width=0.45\textwidth]{figures/V_6_nodes_numbered.pdf}\hspace{0.01cm}
		\includegraphics[width=0.45\textwidth]{figures/GV_6_nodes_numbered.pdf}
	\fi
	\caption{Voronoi diagram (left) for 6 dimensionless nodes, with its equivalent Guaranteed Voronoi diagram (right) in the case of disks (centered on those points).}
	\label{fig:V_GV_comparison}
\end{figure}

\subsection{Guaranteed Voronoi Diagram}
\label{sec:gvoronoi}

The Guaranteed Voronoi (GV) diagram is defined for a set of uncertain regions $D = \{D_1, D_2, \dots , D_n \}, ~D_i \subset \Omega$. Each uncertain region $D_i$ contains all the possible positions of a point $q_i \in \mathbb{R}^2$ whose localization cannot be measured precisely. Thus the GV diagram is defined as
\begin{eqnarray}
	\nonumber
	V_i^g &=& \left\{q\in\Omega \colon \right. \max \left\|q-q_i\right\|  \leq \min \left\|q-q_j\right\| \\
	&&\forall j\in  I_n,\left. ~j\neq i, ~q_i \in D_{i},~q_j \in D_{j}\right\},~i\in I_n.
	\label{eq:GV_definition_max_min}
\end{eqnarray}
An equivalent definition of the GV diagram is as follows
\begin{equation*}
	V_i^g = \bigcap_{i \in I_n} H_{ij}
\end{equation*}
where
\small
\begin{equation*}
	H_{ij} = \left\{ q\in\Omega\colon \max\left\|q-q_i\right\|
	\leq \min\left\|q-q_j\right\|, ~\forall q_i \in D_i, ~\forall q_j \in D_j
	\right\}.
\end{equation*}
\normalsize
This definition is analogous to the construction of the classic Voronoi cell by the halfplane intersection method described previously, with the difference that the regions being intersected are no longer halfplanes.

As such, each cell $Vi^g$ contains the points on the plane that are closer to the corresponding point $q_i$ of region $D_i$ for all possible configurations of the uncertain points. A Guaranteed Voronoi diagram of 6 regions constrained in a subset $\Omega$ of $\mathbb{R}^2$ is shown in Figure \ref{fig:V_GV_comparison} (right), where similarly the cell boundaries are shown in blue.

\begin{remark}
In contrast to the classic Voronoi diagram, the Guaranteed Voronoi diagram does not constitute a tessellation of $\Omega$, since $\bigcup_{i \in I_n} V_i^g \subset \Omega$. We define a neutral region $\mathcal{O}$ so that $\mathcal{O} \cup \bigcup_{i \in I_n} V_i^g = \Omega$.
\label{rem:neutral}
\end{remark}

Similarly to the Voronoi diagram, we can define the Guaranteed Delaunay neighbors $N_i^g$ of a region $D_i$ so that in order to construct $V_i^g$, only the regions in $N_i^g$ need be considered. Thus the Guaranteed Delaunay neighbors are defined as
\begin{equation}
	N_i^g = \left\{ j \in I_n, j \neq i \colon \partial H_{ij} \cap \partial V_i^g \neq \emptyset \right\}.
	\label{eq:guaranteed_delaunay}
\end{equation}

Having defined the Guaranteed Delaunay neighbors, the definition of the GV diagram can be simplified into
\begin{equation}
	V_i^g = \bigcup_{i \in N_i} H_{ij}.
\label{eq:GV_partitioning}
\end{equation}

It is important to note that when two or more uncertain regions overlap, none of them are assigned GV cells. In addition, if the uncertain regions degenerate into points, the GV diagram converges to the classic Voronoi diagram.

\subsection{Guaranteed Voronoi Diagram of Disks}

The GV diagram is examined in the case when the uncertain regions $D_i$ are disks $C_i^u$ as defined in Equation (\ref{eq:positioning_uncertainty}). In that case it has been shown that $H_{ij}$ are regions on the plane bounded by hyperbolic branches. For any two nodes $i$ and $j$, $\partial H_{ij}$ and $\partial H_{ji}$ are the two branches of a hyperbola with foci at $q_i$ and $q_j$ and semi-major axis equal to the sum of the radii of $C_i^u$ and $C_j^u$. The parametric equation of these hyperbola branches can be found in the Appendix.

The dependence of the GV cells of two disks on the distance of their centers $d_{ij} = d(q_i,q_j)$ is as follows. When the disks $C_i^u, C_j^u$, overlap both of the cells $V_i^g, V_j^g$ are empty as shown in Fig. \ref{fig:cells_distance} (a). If the disks represent the possible positions of the agents, this is an undesirable configuration as it can lead to collisions. When the disks $C_i^u, C_j^u$ are outside tangent (i.e., $d(q_i,q_j) = r_i^u+r_j^u$), the resulting cells are rays starting from the centers of the disks $q_i, q_j$ and extending along the direction of $q_iq_j$ as seen in Fig \ref{fig:cells_distance} (b). When the disks are disjoint, the GV cells are bounded by the two branches of a hyperbola. As $d_{ij}$ increases further, the eccentricity of the hyperbola increases and the distance of the disk centers from the hyperbola vertices (the points of the hyperbola closest to its center) increases as seen in Fig \ref{fig:cells_distance} (c), (d). The distance between the hyperbola's vertices remains constant at $r_i+r_j$ as $d_{ij}$ increases.
	
The GV cells of two disks also depend on the sum of their radii $r_i + r_j$ as seen in Figure \ref{fig:cells_dimensions}. As $r_i + r_j$ decreases, the eccentricity of the hyperbola increases and the result on the cells is the same as if the distance of the disks' centers increased, as explained previously. When $r_i + r_j = 0$, the GV cells are the classic Voronoi cells.

\begin{figure}[htb]
	\centering
	\ifx\singlecol\undefined
		\subfloat[]{
		\includegraphics[width=0.23\textwidth]{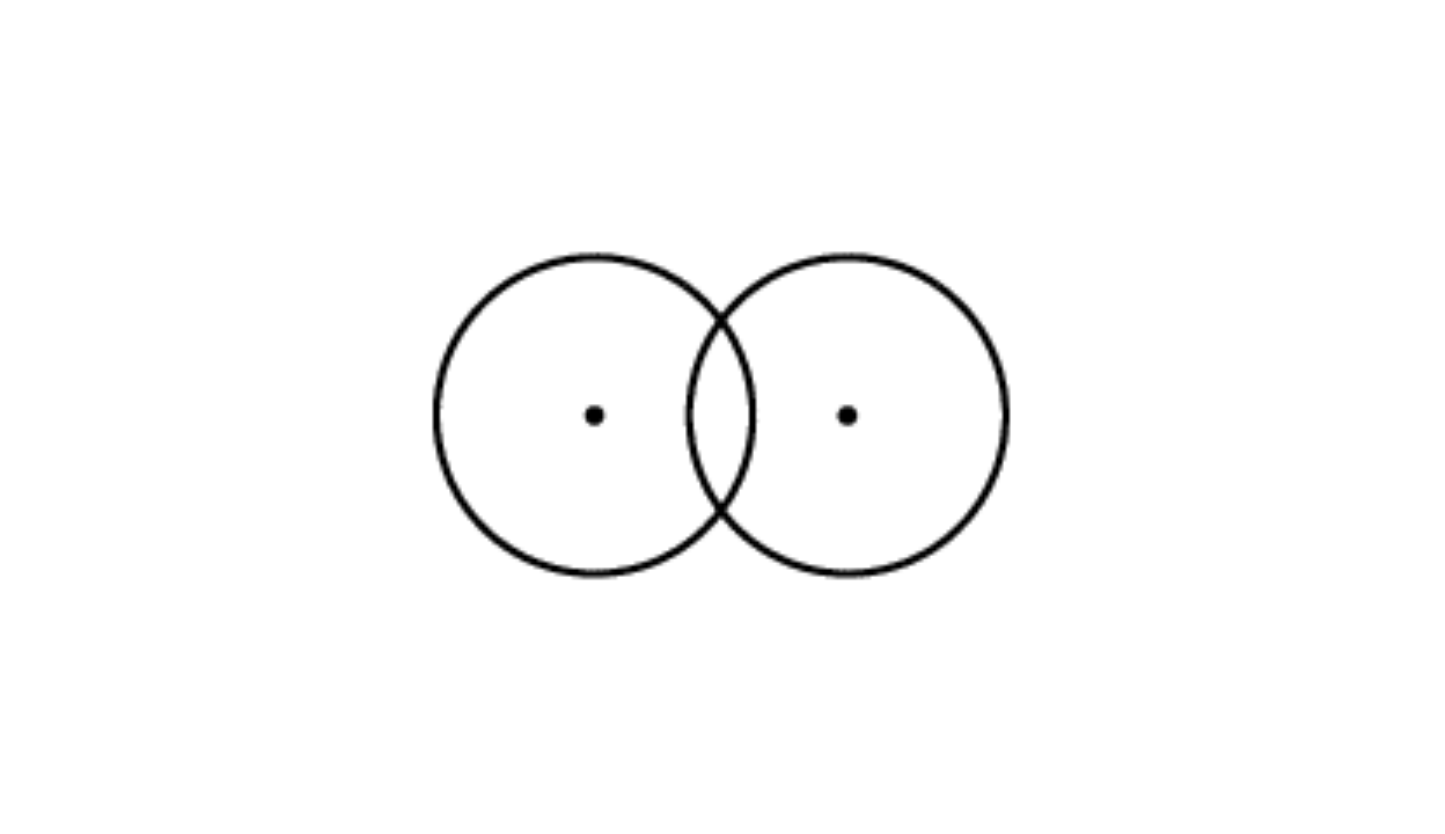} }
		\subfloat[]{ \includegraphics[width=0.23\textwidth]{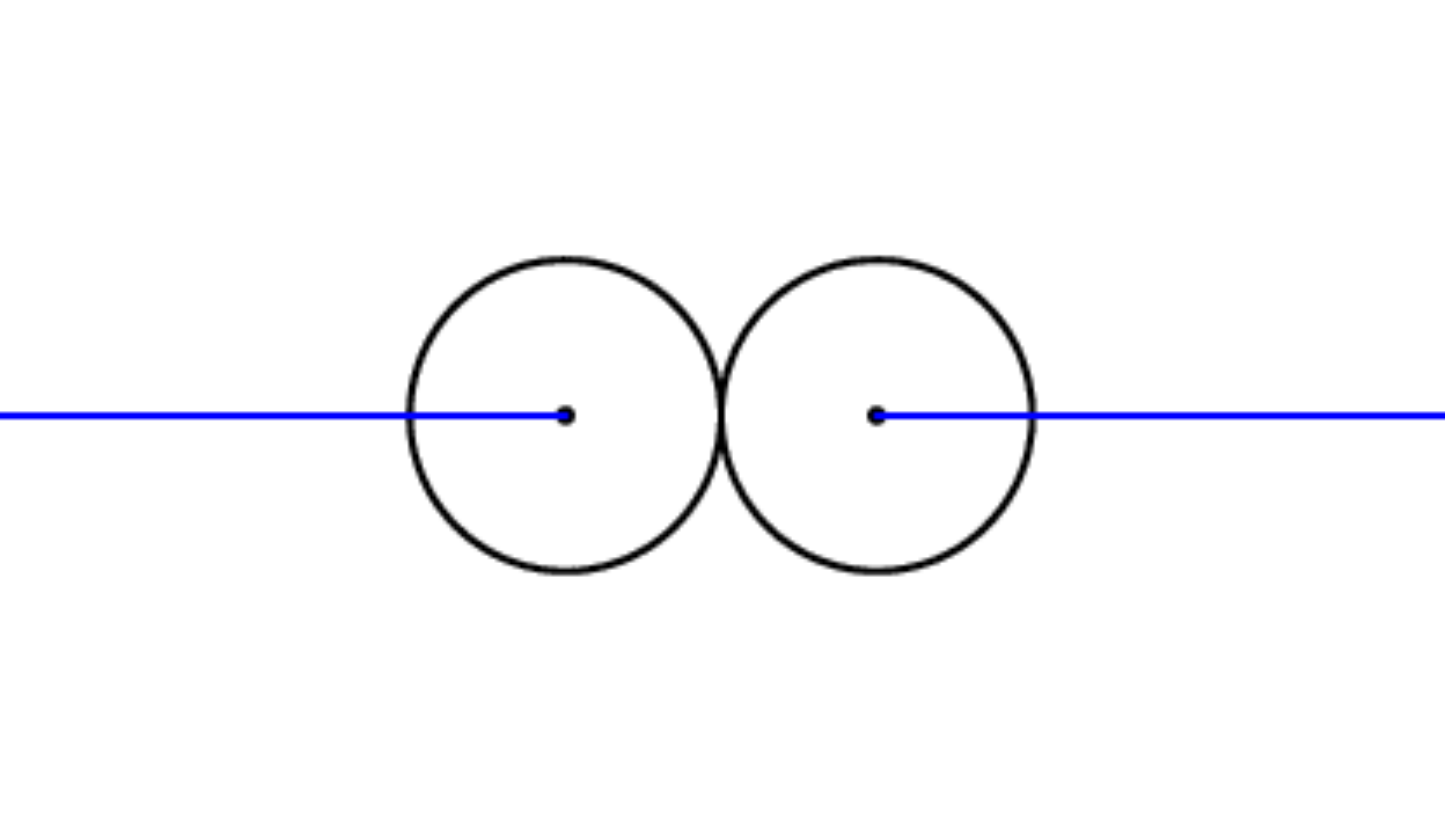} }\\
		\subfloat[]{ \includegraphics[width=0.23\textwidth]{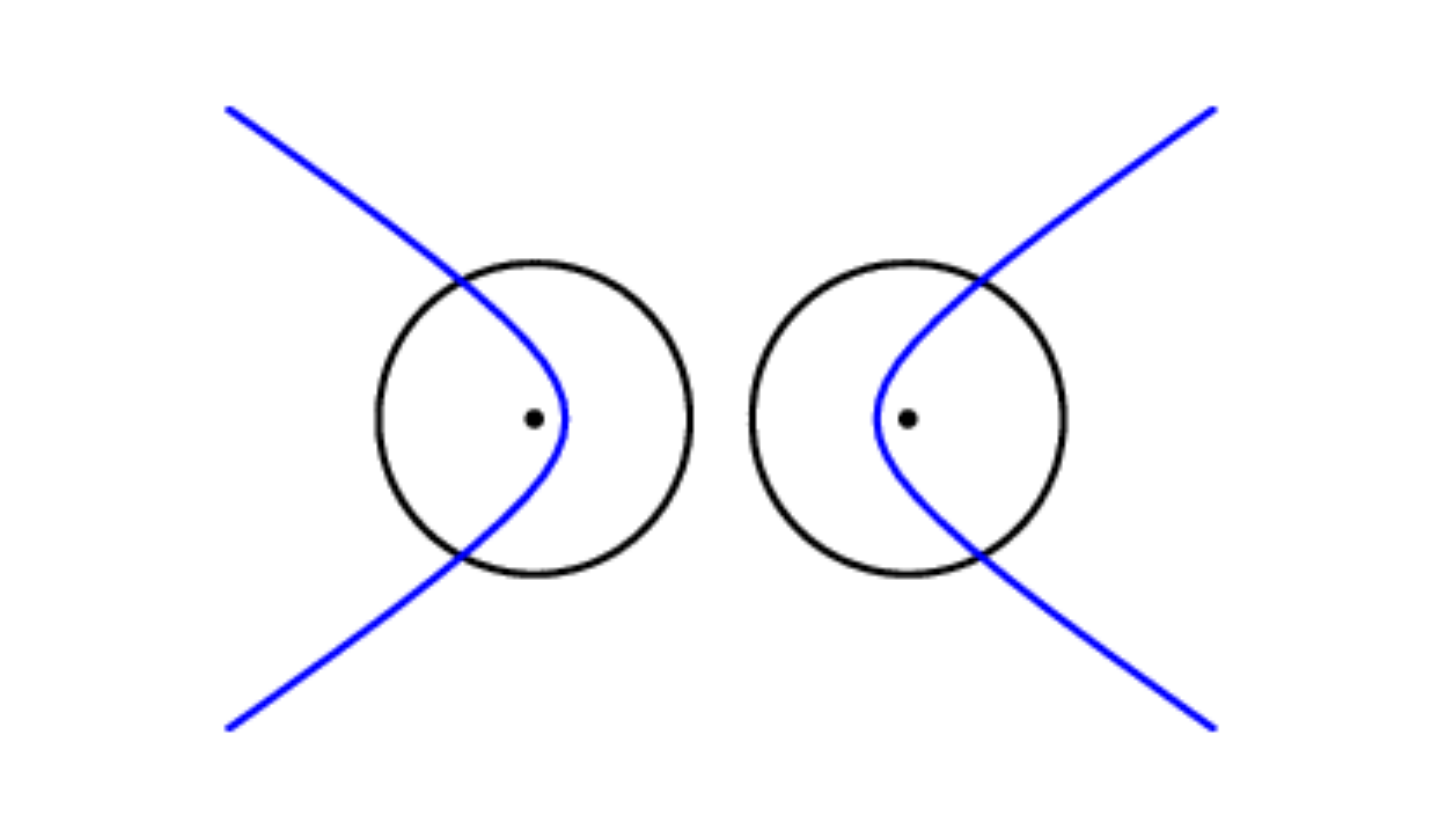} }
		\subfloat[]{ \includegraphics[width=0.23\textwidth]{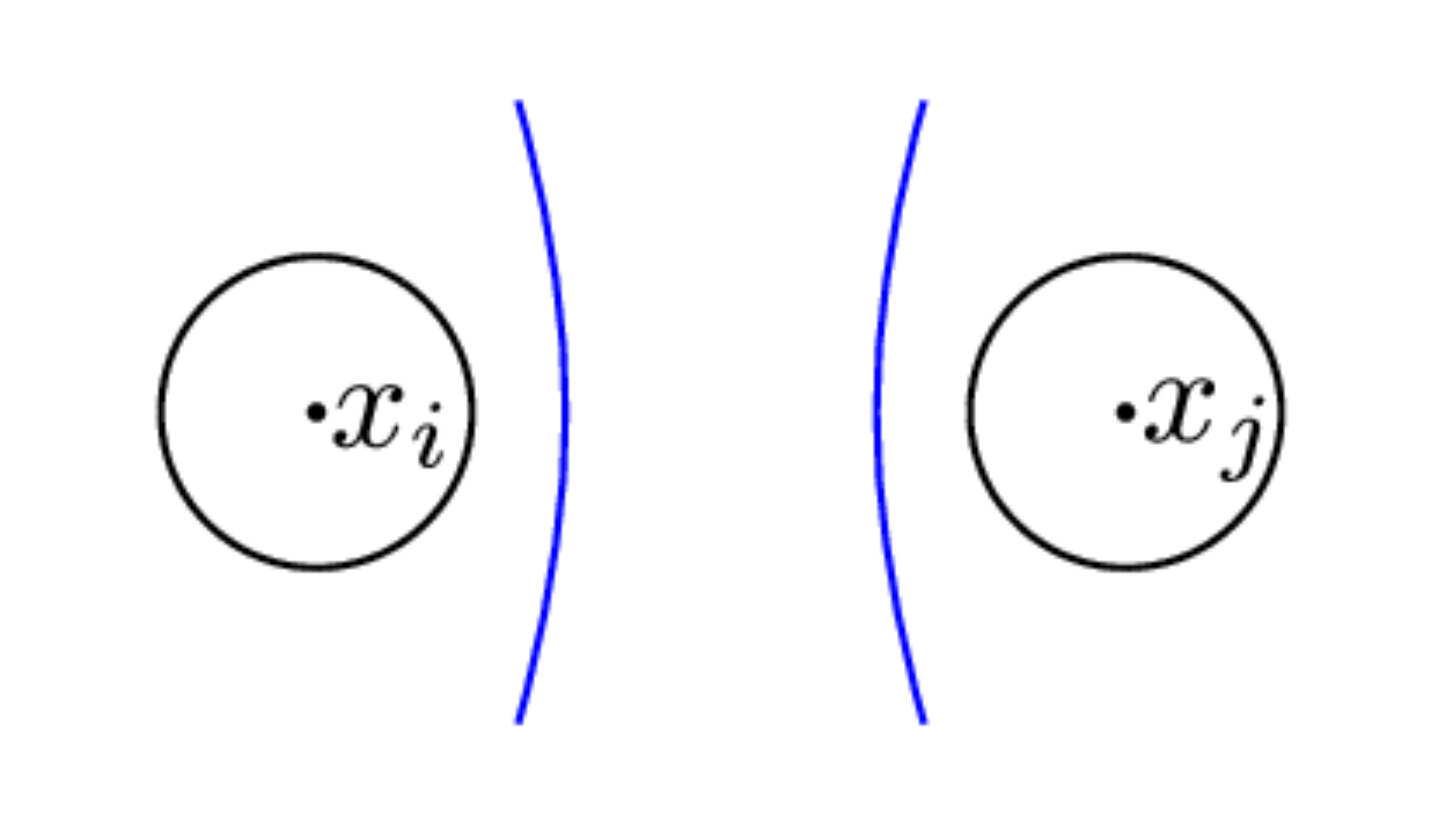} }
	\else
		\subfloat[]{
		\includegraphics[width=0.45\textwidth]{figures/GV_cells_2_nodes/1_overlap_es.pdf} }
		\subfloat[]{ \includegraphics[width=0.45\textwidth]{figures/GV_cells_2_nodes/2_tangent_es.pdf} }\\
		\subfloat[]{ \includegraphics[width=0.45\textwidth]{figures/GV_cells_2_nodes/3_close_es.pdf} }
		\subfloat[]{ \includegraphics[width=0.45\textwidth]{figures/GV_cells_2_nodes/4_far_text_es.pdf} }
	\fi
	
	\caption{Dependence of the GV cells on $\left\|q_i-q_j\right\|$.}
	\label{fig:cells_distance}
\end{figure}

\begin{figure}[htb]
	\centering
	\ifx\singlecol\undefined
		\includegraphics[width=0.155\textwidth]{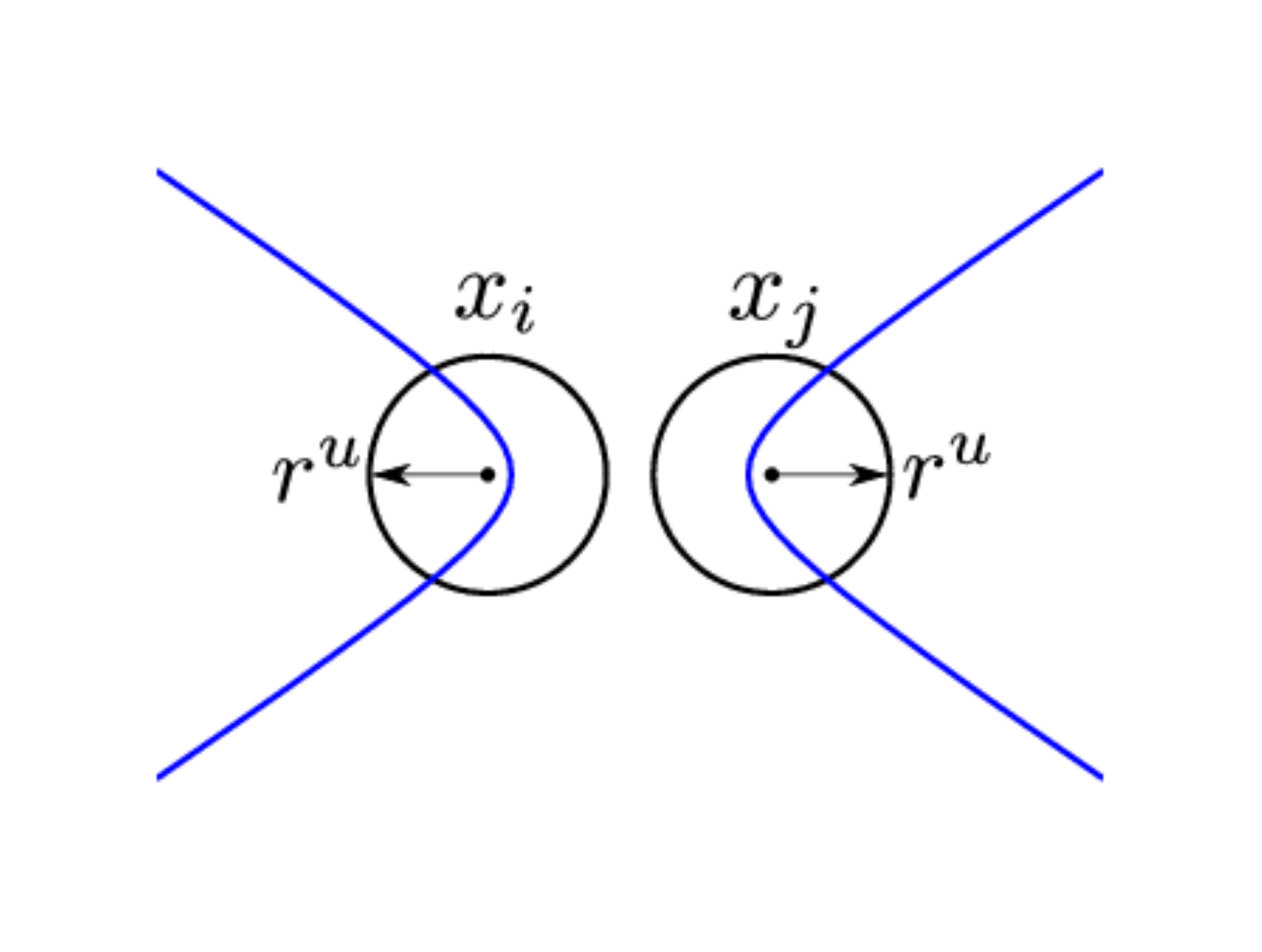}
		\includegraphics[width=0.155\textwidth]{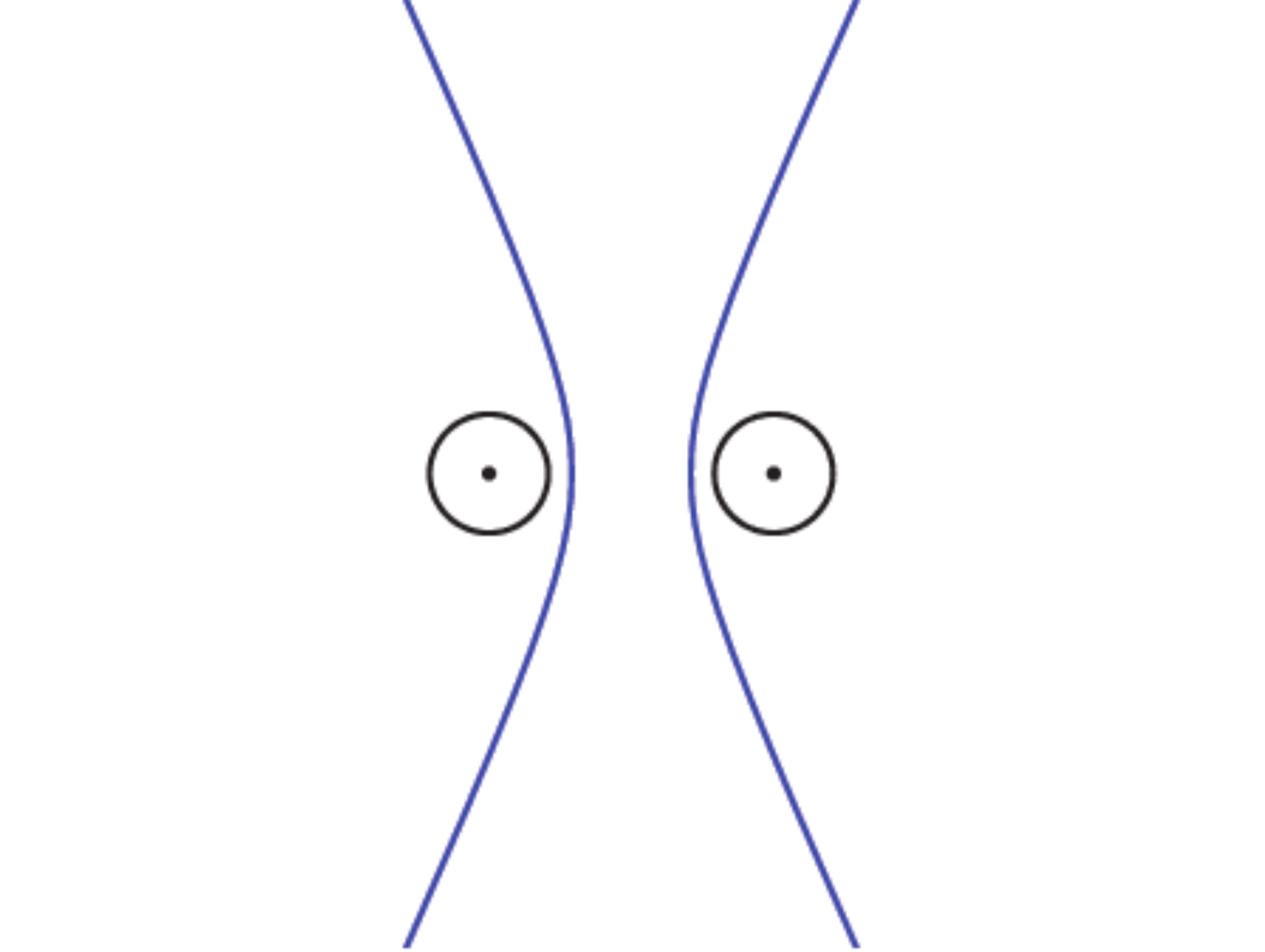}
		\includegraphics[width=0.155\textwidth]{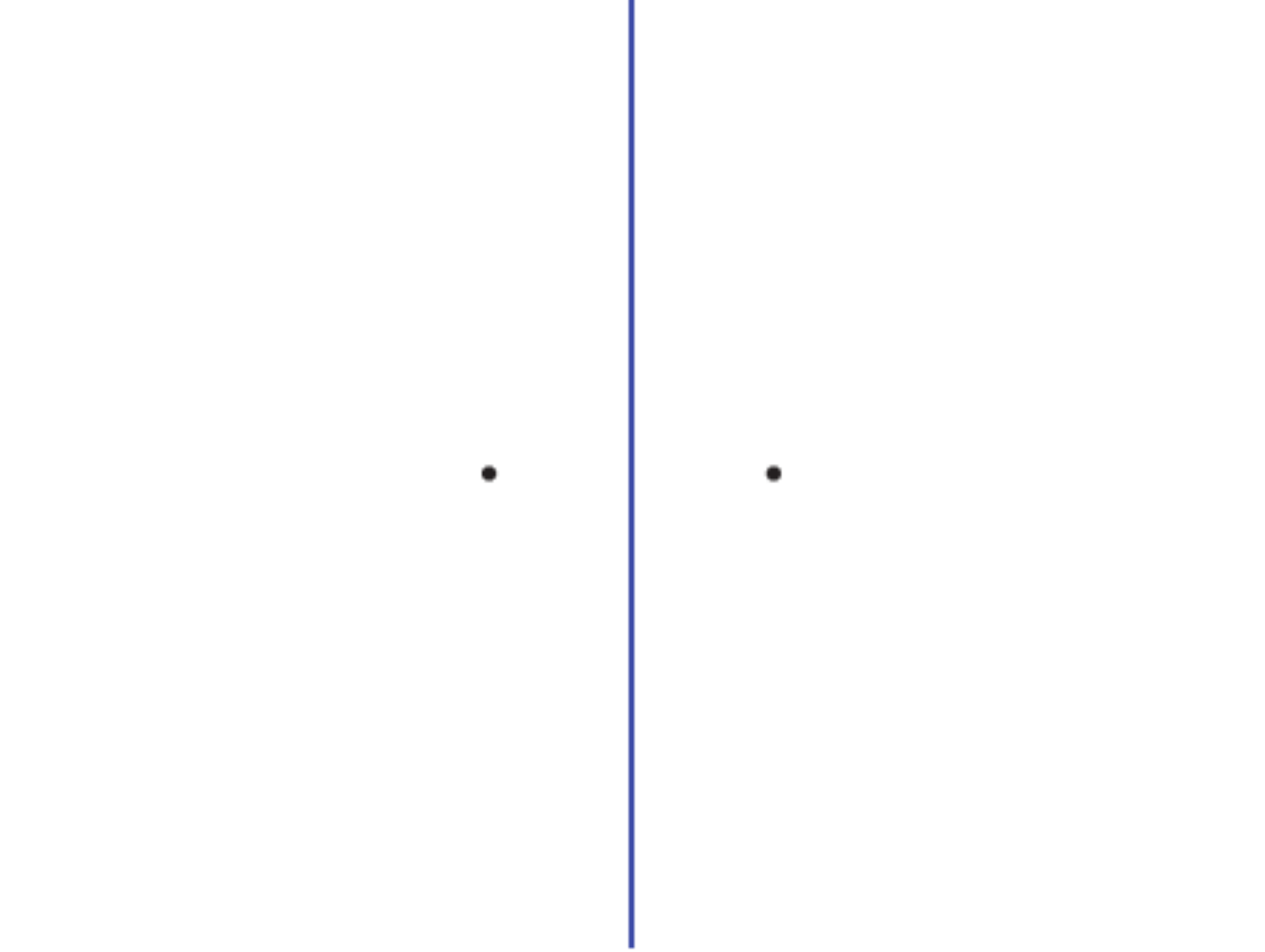}
	\else
		\includegraphics[width=0.3\textwidth]{figures/GV_cells_2_nodes_radii/1-large_text_es.pdf}
		\includegraphics[width=0.3\textwidth]{figures/GV_cells_2_nodes_radii/2-small_es.pdf}
		\includegraphics[width=0.3\textwidth]{figures/GV_cells_2_nodes_radii/3-zero_es.pdf}
	\fi
	
	\caption{Dependence of the GV cells on $r_i + r_j$.}
	\label{fig:cells_dimensions}
\end{figure}

It has been also shown that the Guaranteed Delaunay neighbors $N_i^g$ are a subset of the classic Delaunay neighbors of the uncertain disks' centers. The two sets $N_i$ and $N_i^g$ are equal when the GV diagram is constructed for the whole plane, however when the GV diagram is constrained within a subset $\Omega$ of the plane, $N_i^g$ can be a subset of $N_i$. Such a case is shown in Figure \ref{fig:V_GV_comparison} where nodes 1 and 5 are Delaunay neighbors but not Guaranteed Delaunay neighbors. This way, the GV diagram of disks can be constructed in a manner similar to the classic Voronoi diagram. First the Delaunay triangulation of the disks' centers is constructed and then each cell $V_i^g$ is constructed by intersecting the regions bound by hyperbolic arcs $H_{ij} \forall j \in N_i^g$. This process is shown in Figure \ref{fig:GV_construction}.

\begin{figure}[htb]
	\centering
	\ifx\singlecol\undefined
		\includegraphics[width=0.4\textwidth]{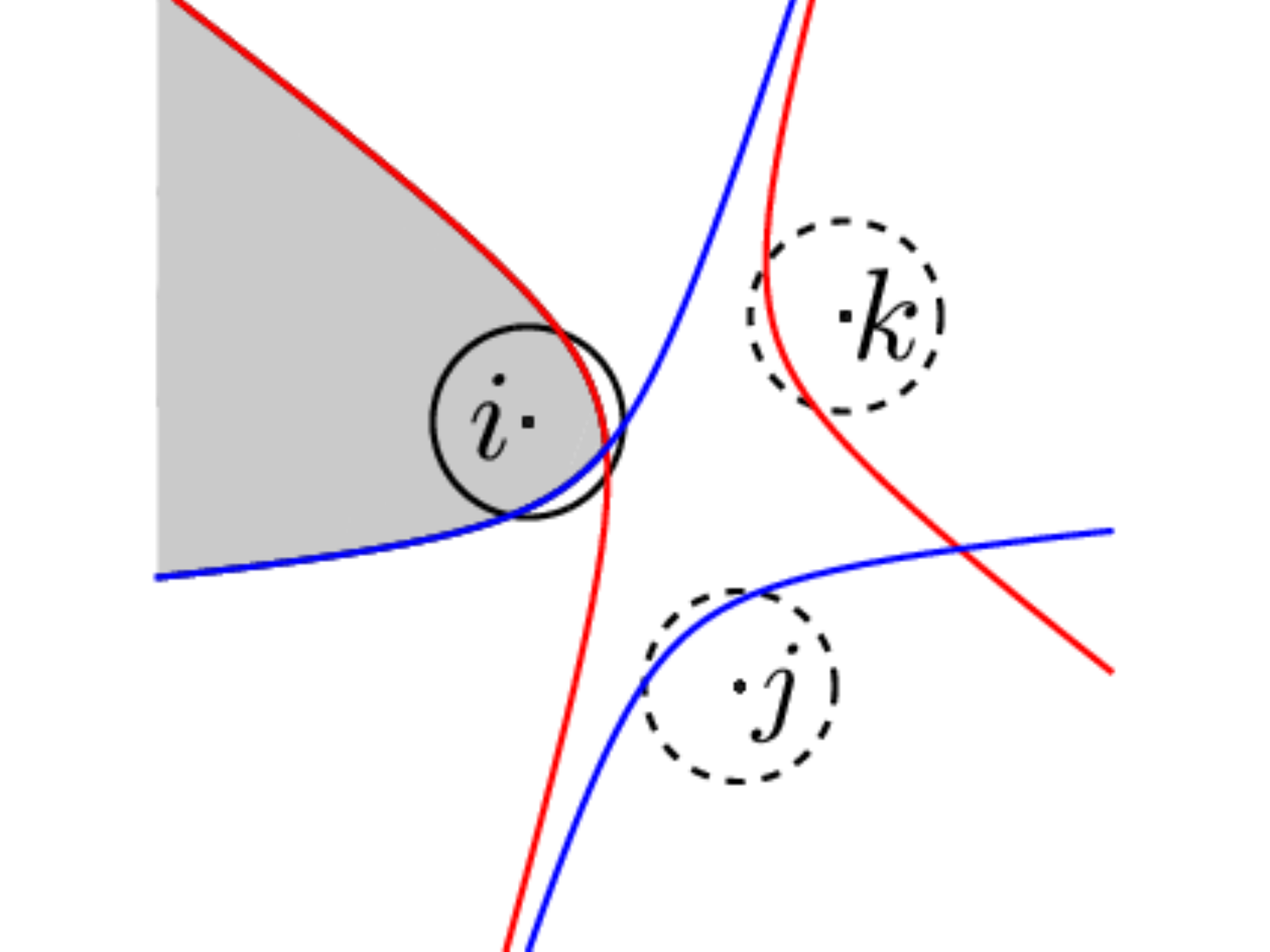}
	\else
		\includegraphics[width=0.8\textwidth]{figures/GV_construction_text.pdf}
	\fi
	\caption{The construction of a GV cell by the hyperbolic region intersection method.}
	\label{fig:GV_construction}
\end{figure}

It is useful to define Guaranteed-Sensing Voronoi (GSV) cells as
\begin{equation}
	V_i^{gs}=V_i^g\cap C_{i}^{gs},~~i\in I_n,
\label{eq:guaranteed sensing cells}
\end{equation}
which are the parts of the GV cells that are guaranteed to be sensed by their respective nodes. It can be shown that since GV cells are disjoint, GSV cells are also disjoint regions.

\section{Problem Statement - Distributed Control Law}
\label{sec:problem_law}
By taking into account the GV partitioning of the space, the following area coverage criterion is constructed
\begin{equation}
	\mathcal{H} =
	\int_{\bigcup_{i\in I_n}V_i^{gs}} \phi\left(q\right)\mathrm{d}q=
	\sum\limits_{i\in I_n}\int_{V_i^{gs}}\phi\left(q\right)\mathrm{d}q=
	\sum\limits_{i\in I_n}\mathcal{H}_{i},
	\label{coverage_criterion}
\end{equation}
where the function $\phi\colon \mathbb{R}^2\rightarrow \mathbb{R}_+$ is related to the a~priori knowledge of importance of a point $q\in\Omega$ indicating the probability of an event to take place at $q$ in a coverage scenario. This criterion expresses the area that is guaranteed to be closest to and covered by the nodes. 

A distributed control law is designed in order to maximize the coverage criterion (\ref{coverage_criterion}) while taking into account the nodes' guaranteed sensing pattern (\ref{eq:guaranteed_sensing}) and the GV partitioning (\ref{eq:GV_partitioning}). It is shown that this control law leads to a monotonic increase of the coverage criterion.

\begin{theorem}
For a network of mobile nodes with uncertain localization as in (\ref{eq:positioning_uncertainty}), uniform range--limited radial sensing performance as in (\ref{eq:guaranteed_sensing}), governed by the individual agent's kinodynamics described in (\ref{eq:dynamics}), and the GV--partitioning of $\Omega$ defined in (\ref{eq:GV_partitioning}), the coordination scheme
\begin{eqnarray}
	\nonumber
	u_i &=& \alpha_i \int_{\partial V_i^{gs} \cap\partial C_i^{gs}}n_i\,\phi\,\mathrm{d}q
	+ 
	\\
	&& \hspace{-2cm} \alpha_i  \sum_{j\in N_i^g}\left[\int_{\partial V_i^{gs} \cap \partial H_{ij}} \upsilon_i^i\,n_i\,\phi\,\mathrm{d}q + \int_{\partial V_j^{gs} \cap \partial H_{ji}} \upsilon_j^i\,n_j\,\phi\,\mathrm{d}q\right]
	\label{initial_control_law}
\end{eqnarray}
where $n_i$ is the outward unit normal on $\partial V_i^{gs}$ and $\alpha_i$ a positive constant, maximizes the performance criterion (\ref{coverage_criterion}) along the nodes' trajectories in a monotonic manner, leading in a locally area--optimal configuration of the network.
\label{thm:initial_law}
\end{theorem}

\begin{proof}
In order to guarantee monotonous increase of the coverage criterion (\ref{coverage_criterion}), we evaluate its time derivative
\begin{equation}
\frac{\partial \mathcal{H}}{\partial t}  = \sum_{i \in I_n} \frac{\partial \mathcal{H}_i}{\partial q_i} \dot{q}_i = \sum_{i \in I_n} \frac{\partial \mathcal{H}_i}{\partial q_i} u_i
\end{equation}

By using the gradient--based control law
\begin{equation}
u_i = \frac{\partial \mathcal{H}_i}{\partial q_i}, ~i \in I_n,
\end{equation}
monotonous increase of $\mathcal{H}$ is achieved.
 
We now evaluate the partial derivative of $\mathcal{H}$ with respect to $q_i$ as
\begin{equation*}
\frac{\partial\mathcal{H}}{\partial q_i}=
\frac{\partial}{\partial q_i}\int_{V_i^{gs}}\phi  dq+ \sum_{j \in N_i^g}\frac{\partial}{\partial q_i}\int_{V_j^{gs}}\phi dq.
\end{equation*}
Considering the second summation term, infinitesimal motion of $q_i$ may only affect $\partial V_j^{gs}$ at $\partial V_j^{gs} \cap \partial \mathcal{O}$, where $\mathcal{O}$ is the neutral region defined in Remark \ref{rem:neutral}, since both hyperbola branches are affected by alteration of one of the foci. In addition, only the Guaranteed Delaunay neighbors (\ref{eq:guaranteed_delaunay}) of $i$ are considered in the summation, as a major property of GV--partitioning. Therefore, the former expression can be written via the generalized Leibniz integral rule \cite{Flanders_AMM73} (by converting surface integrals to line ones) as
\begin{equation*}
\frac{\partial \mathcal{H}}{\partial q_i}=
\int_{\partial V_i^{gs}}\upsilon_i^i\,n_i\,\phi dq 
+ 
\sum_{j \in N_i^g}\int_{\partial V_j^{gs} \cap \partial \mathcal{O}} \upsilon_j^i\,n_j\,\phi dq,
\end{equation*}
where $\upsilon_i^i,\upsilon_j^i$ stand for the transpose Jacobian matrices with respect to $q_i$ of the points $q\in \partial V_i^{gs}$, $q\in \partial V_j^{gs}$, respectively, i.e.
\begin{equation}
	\upsilon_j^i\left(q\right) \stackrel{\triangle}{=} { \frac{\partial q}{\partial q_i} }^T,~~q\in \partial V_j^{gs},~i,j\in I_n.
	\label{ch3:eq:upsilon tilde}
\end{equation}
The boundary $\partial \tilde{V}_i^{gs}$ can be decomposed in disjoint sets as
\begin{equation}
\hspace{-0.05cm}\partial V_i^{gs}=
\left\{\partial V_i^{gs}\cap\partial\Omega\right\}
\cup 
\left\{\partial V_i^{gs}\cap\partial C_i^{gs}\right\}
\cup
\left\{\partial V_i^{gs} \cap \partial \mathcal{O} \right\}.
\end{equation}
These sets represent the parts of $\partial V_i^{gs}$ that lie on the boundary of $\Omega$, the boundary of the node's guaranteed sensing region, and the boundary of the unassigned neutral region of $\Omega$, respectively. This decomposition is shown in Figure \ref{fig:boundary_decomposition} with $\partial V_i^{gs}\cap\partial\Omega$ in green, $\partial V_i^{gs}\cap\partial C_i^{gs}$ in solid red and $\partial V_i^{gs} \cap \partial \mathcal{O}$ in solid blue.

\begin{figure}[htb]
	\centering
	\ifx\singlecol\undefined
		\includegraphics[width=0.4\textwidth]{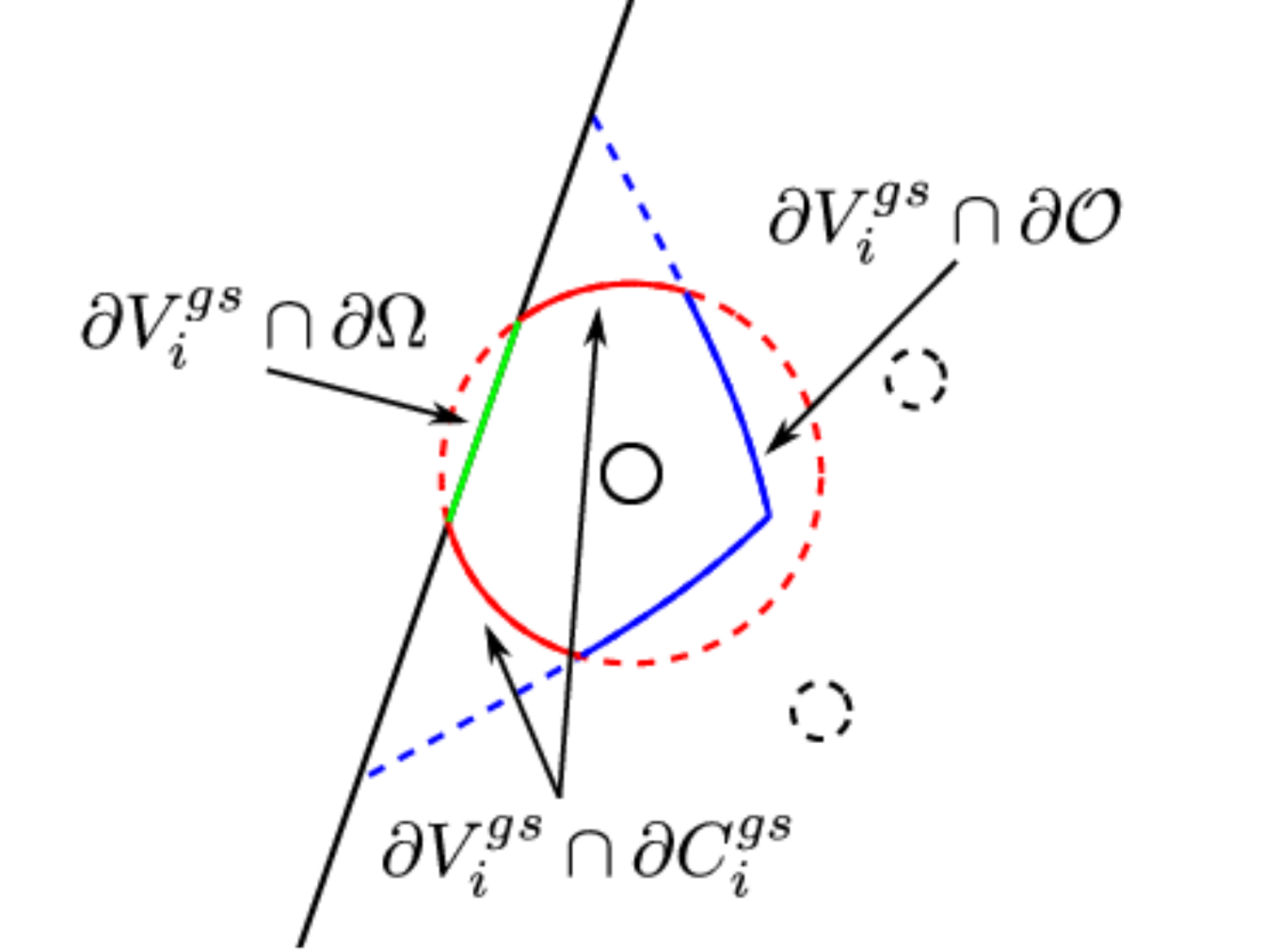}
	\else
		\includegraphics[width=0.8\textwidth]{figures/boundary_decomposition_text.pdf}
	\fi
	
	\caption{The decomposition of $\partial V_i^{gs}$ into disjoint sets (solid red, green and blue).}
	\label{fig:boundary_decomposition}
\end{figure}

Hence $\frac{\partial \mathcal{H}}{\partial q_i}$ can be written as
\begin{eqnarray}
		\frac{\partial \mathcal{H}}{\partial q_i} &=&
\int_{\partial V_i^{gs} \cap \partial \Omega}\upsilon_i^i\,n_i\,\phi\,\mathrm{d}q
+
\int_{\partial V_i^{gs} \cap \partial C_i^{gs}}\upsilon_i^i\,n_i\,\phi\,\mathrm{d}q
+
\nonumber
\\
&& \hspace{-1cm}
\int_{\partial V_i^{gs} \cap \partial \mathcal{O}} \upsilon_i^i\,n_i\,\phi\,\mathrm{d}q 
+
\sum_{j \in N_i^g}\int_{\partial V_j^{gs} \cap \partial \mathcal{O}} \upsilon_j^i\,n_j\,\phi\,\mathrm{d}q.
\end{eqnarray}
It is apparent that $\upsilon_i^i=0$ at $q\in \partial V_i^{gs}\cap\Omega$ since all $q\in\partial \Omega$ remain unaltered by infinitesimal motions of $q_i$, assuming no alteration of the environment. 
Considering the second integral, for any point $q \in \partial\tilde{V}_i^{gs}\cap\partial C_i^{gs}$ it holds that $\upsilon_i^i\left(q\right)=\mathbb{I}_2$, where $\mathbb{I}$ stands for the identity matrix, since they translate along the direction of motion of $q_i$ at the same rate. 
As far as the sets $\partial V_i^{gs} \cap \partial \mathcal{O}$ and $\partial V_j^{gs} \cap \partial \mathcal{O},~j\in N_i^g$ are concerned, they can be merged in pairs via utilization of the left and right hyperbolic branches, as introduced in GV--partitioning (\ref{eq:GV_partitioning}), as follows
\begin{eqnarray}
\nonumber
\int_{\partial V_i^{gs} \cap \partial \mathcal{O}} \upsilon_i^i\,n_i\,\phi\,\mathrm{d}q 
+
\sum_{j \in N_i^g}\int_{\partial V_j^{gs} \cap \partial \mathcal{O}} \upsilon_j^i\,n_j\,\phi\,\mathrm{d}q=
\\
\sum_{j \in N_i^g}\left[\int_{\partial V_i^{gs} \cap \partial H_{ij}} \upsilon_i^i\,n_i\,\phi\,\mathrm{d}q + \int_{\partial V_j^{gs} \cap \partial H_{ji}} \upsilon_j^i\,n_j\,\phi\,\mathrm{d}q\right]~.
\end{eqnarray}

However, for any two Delaunay neighbors $i,j$ it can be observed that 
\begin{itemize}
	\item The hyperbolic branches $\partial H_{ij}$ and $\partial H_{ji}$ are symmetric with respect to the perpendicular bisector of $q_i,q_j$, and are governed by the same set of parametric equations (left and right branch).
	\item The vectors $n_j$ are mirrored images of the corresponding $n_i$ with respect to the perpendicular bisector of $q_i,q_j$.
\end{itemize}

Taking into account the above, $\frac{\partial\mathcal{H}}{\partial q_i}$ can be written as
\begin{eqnarray}
\nonumber
\frac{\partial\mathcal{H}}{\partial q_i}&=&\int_{\partial V_i^{gs} \cap\partial C_i^{gs}}n_i\,\phi\,\mathrm{d}q
+ 
\\
&& \hspace{-2cm}\sum_{j \in N_i^g}\left[\int_{\partial V_i^{gs} \cap \partial H_{ij}} \upsilon_i^i\,n_i\,\phi\,\mathrm{d}q + \int_{\partial V_j^{gs} \cap \partial H_{ji}} \upsilon_j^i\,n_j\,\phi\,\mathrm{d}q\right]~.
\label{criterion_derivative}
\end{eqnarray}

The unit normal vectors $n_i,n_j$ and Jacobian matrices $\upsilon_i^i,\upsilon_j^i$ in the second part of (\ref{criterion_derivative}) can be evaluated via utilization of the parametric representations of the sets over which the integration takes place. These sets are parts of the left ($\partial H_{ij}$) and right ($\partial H_{ji}$) branch of the hyperbola that assigns the space among two arbitrary nodes $i,j$. The second term of the sum in (\ref{criterion_derivative}) requires knowledge of all the nodes in $\bigcup_{j \in N_i^g} N_j^g \subseteq I_n$ and so the control law is distributed. 

The decomposition of the set $\partial V_i^{gs} \cap \partial \mathcal{O}$ of node $i$ into mutually disjoint hyperbolic arcs $\partial V_i^{gs} \cap \partial H_{ij}$ and $\partial V_i^{gs} \cap \partial H_{ik}$ is shown in Figure \ref{fig:control_decomposition}. Figure \ref{fig:control_decomposition} also shows the other domains of integration used in (\ref{criterion_derivative}) $\partial V_i^{gs} \cap \partial C_{i}^{gs}$, $\partial V_j^{gs} \cap \partial H_{ji}$ and $\partial V_k^{gs} \cap \partial H_{ki}$.

\begin{figure}[htb]
	\centering
	\ifx\singlecol\undefined
		\includegraphics[width=0.4\textwidth]{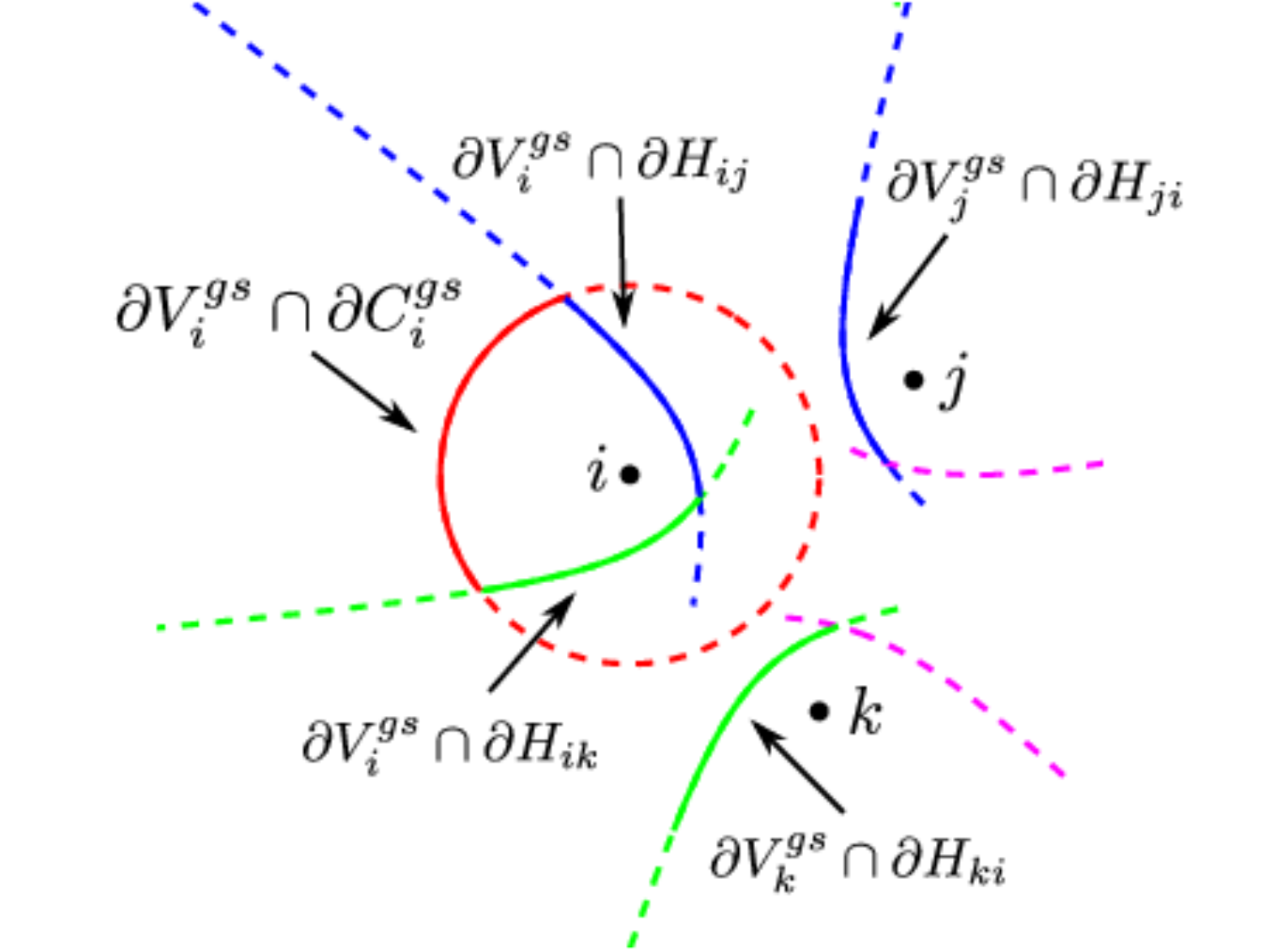}
	\else
		\includegraphics[width=0.8\textwidth]{figures/control_decomposition2_text.pdf}
	\fi
	
	\caption{The domains of integration of the control law (\ref{initial_control_law}).}
	\label{fig:control_decomposition}
\end{figure}

The proposed law (\ref{initial_control_law}) leads to a gradient flow of $\mathcal{H}$ along the nodes trajectories, while $\mathcal{H}$ increases monotonically, since
\[\frac{d\mathcal{H}}{dt}=\sum_{i\in I_n}\frac{\partial\mathcal{H}}{\partial q_i}\cdot \dot{q}_i=\sum_{i\in I_n}\left\|u_i\right\|^2\geq 0.\]
\end{proof}

\begin{remark}
It is not immediately clear from the control law (\ref{initial_control_law}) the information node $i$ requires in order to implement it. The integrals over $\partial V_i^{gs} \cap\partial C_i^{gs}$ and $\partial V_i^{gs} \cap\partial H_{ij}$ require knowledge of the Guaranteed Delaunay neighbors $N_i^g$ in order to compute $V_i^{gs}$ and $H_{ij}$. However the integral over $\partial V_j^{gs} \cap\partial H_{ji}$ requires knowledge of the whole cell $V_j^{gs}$ of node $i$ and as such requires informations form the Guaranteed Delaunay neighbors $N_j^g$ of node $j$. Thus node $i$ requires information from all nodes in
\begin{equation*}
\tilde{N}_i^g = \bigcup_{j \in N_i^g} N_j^g,
\end{equation*}
which are the Guaranteed Delaunay neighbors of all the Guaranteed Delaunay neighbors of node $i$. 
\label{rem:required_info}
\end{remark}

\subsection{Constraining Nodes Inside $\Omega$}
\label{sec:constrain_movement}

When using this control law, it is possible that a portion of some node's positioning uncertainty region lies outside the region $\Omega$, i.e. that $C_i^u \cap \Omega \neq C_i^u$ for some node $i$. As a result, it would be possible for that node to be outside the region $\Omega$ given that it may be anywhere within $C_i^u$.

In order to prevent this situation, instead of $\Omega$, a subset $\Omega_i^s \subseteq \Omega$ will be used, in the general case different for each node since their uncertainty radii are allowed to differ. This subset is calculated as the Minkowski difference of $\Omega$ with the disk $C\left(r_i^u\right) = \left\{ q \in \Omega \colon \parallel q \parallel \leq r_i^u \right\}$.

\begin{equation*}
\Omega_i^s = \left\{ q \in \Omega ~|~ q + C(r_i^u) \subseteq \Omega \right\}, i \in I_n.
\end{equation*}

Thus by constraining the center of each uncertainty disk $q_i$ inside $\Omega_i^s$, it is guaranteed that no part of the uncertainty disk will ever be outside $\Omega$. This is achieved by projecting the velocity control input $u_i$ on $\partial \Omega_i^s$. This is only needed if $q_i \in \partial \Omega_i^s$ and the contorl input $u_i$ points towards the outside of $\Omega_i^s$. As a result, the control law becomes
\begin{equation}
\tilde{u_i} = 
\begin{cases}
\text{proj}_{\partial \Omega_i^s} u_i & \text{if} q_i \in \partial \Omega_i^s ~\wedge~ q_i+\epsilon~u_i \notin \Omega_i^s \\
u_i & \text{otherwise}
\end{cases}
\label{eq:law_constrain}
\end{equation}
where $\epsilon$ is an infinitesimally small positive constant.

\subsection{Suboptimal Control Law}
\label{sec:control_suboptimal}
Because of the high complexity of the integrals over $H_{ij}$ and $H_{ji}$ of the control law (\ref{initial_control_law}), the large amount of information required in order to implement it, as described in Remark \ref{rem:required_info}, as well as the problems described in Section \ref{sec:constrain_movement}, a simplified version of the control law is proposed.

\begin{conjecture}
The computationally efficient suboptimal control law 
\begin{equation}
u_i = \alpha_i \int_{\partial V_i^{gs} \cap\partial C_i^{gs}}n_i\,\phi\,\mathrm{d}q
\label{eq:control_suboptimal}
\end{equation}
leads the nodes to suboptimal trajectories.
\end{conjecture}

The control law enhancement described in Section \ref{sec:constrain_movement} is used in conjunction with this simplified control law.


\section{Simulation Studies}
\label{sec:simulations}
Simulations have been conducted in order to evaluate the efficiency of the control law as well as compare the complete (\ref{initial_control_law}) with the simplified (\ref{eq:control_suboptimal}) law. All nodes have common positioning uncertainty regions $C_i, \forall q_i \in \Omega$ as well as sensing performance $C_i^{gs}$. The a~priori importance of points inside the region was $\phi(q) = 1, ~\forall q \in \Omega$. In this case, the maximum possible value for $\mathcal{H}$ is
\begin{equation*}
\mathcal{H}^{\max} = \sum_{i \in I_n} \mathcal{A}\left( C_i^{gs} \right)~,
\end{equation*}
where $\mathcal{A}(\cdot)$ is the area function, given that $n$ guaranteed sensing disks can be packed inside $\Omega$. In all simulations, the control law extension as described in Section \ref{sec:constrain_movement} is used.

\subsection{Case Study I}
In this case study the convergence speed between the optimal and suboptimal laws is examined. A network of three nodes in a large region is simulated. Figure \ref{fig:sim1_network} shows the initial and final states of the network, with the uncertainty disks shown in black, the guaranteed sensing disks in red and the GV cells in blue. In Figure \ref{fig:sim1_traj} the trajectories of the nodes when using the optimal control law are shown in red and when using the suboptimal control law in blue, while the initial positions are represented by circles and the final by squares. It is observed that the final configurations in both cases are very similar since the average distance between nodes' final states is $0.57 \%$ of the region's diameter. This is due to the low number of nodes in the network and is a result that can not be generalized, as seen in the following case study. Figure \ref{fig:sim1_H} shows the evolution of the coverage objective $\mathcal{H}$ with time, again shown in red for the optimal control law and in blue for the suboptimal. As it is expected, the optimal law converges faster than the suboptimal. The final value of the objective function is the same for both control laws in this case since the network reaches a globally optimal state in both cases. 

\begin{figure}[htbp]
	\centering
	\ifx\singlecol\undefined
		\subfloat[]{ \includegraphics[width=0.3\textwidth]{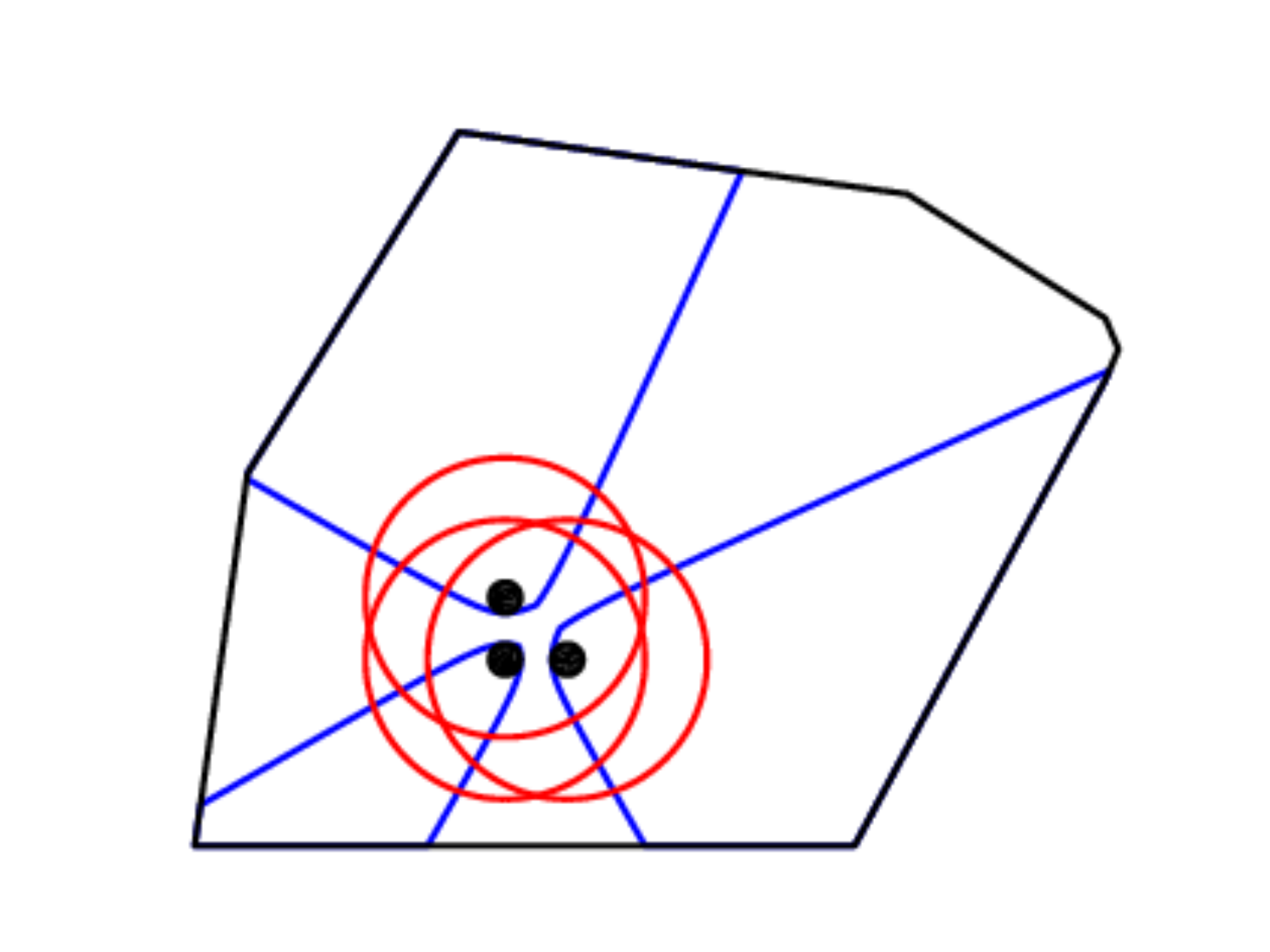} }\\
		\subfloat[]{ \includegraphics[width=0.23\textwidth]{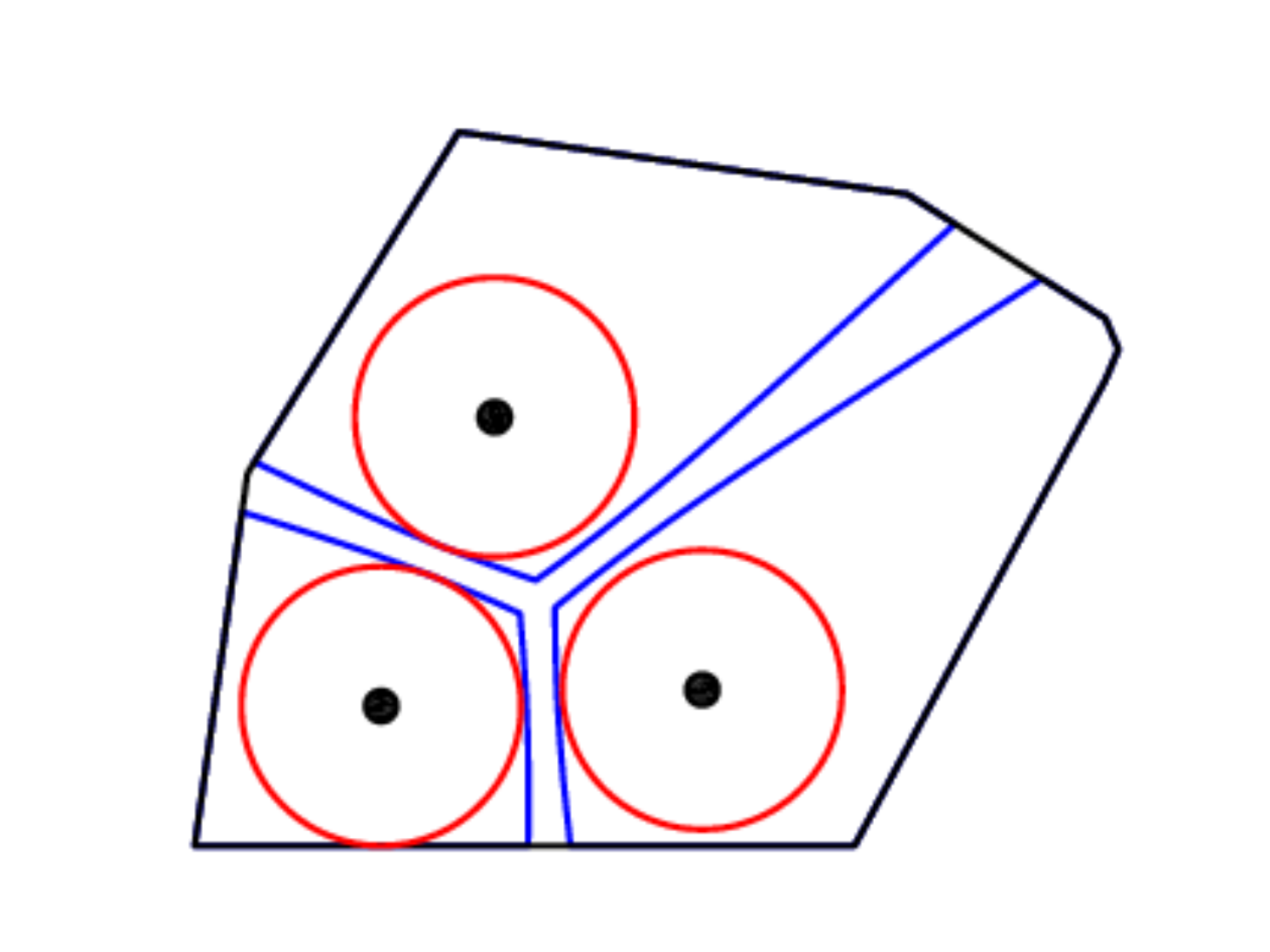} }
		\subfloat[]{ \includegraphics[width=0.23\textwidth]{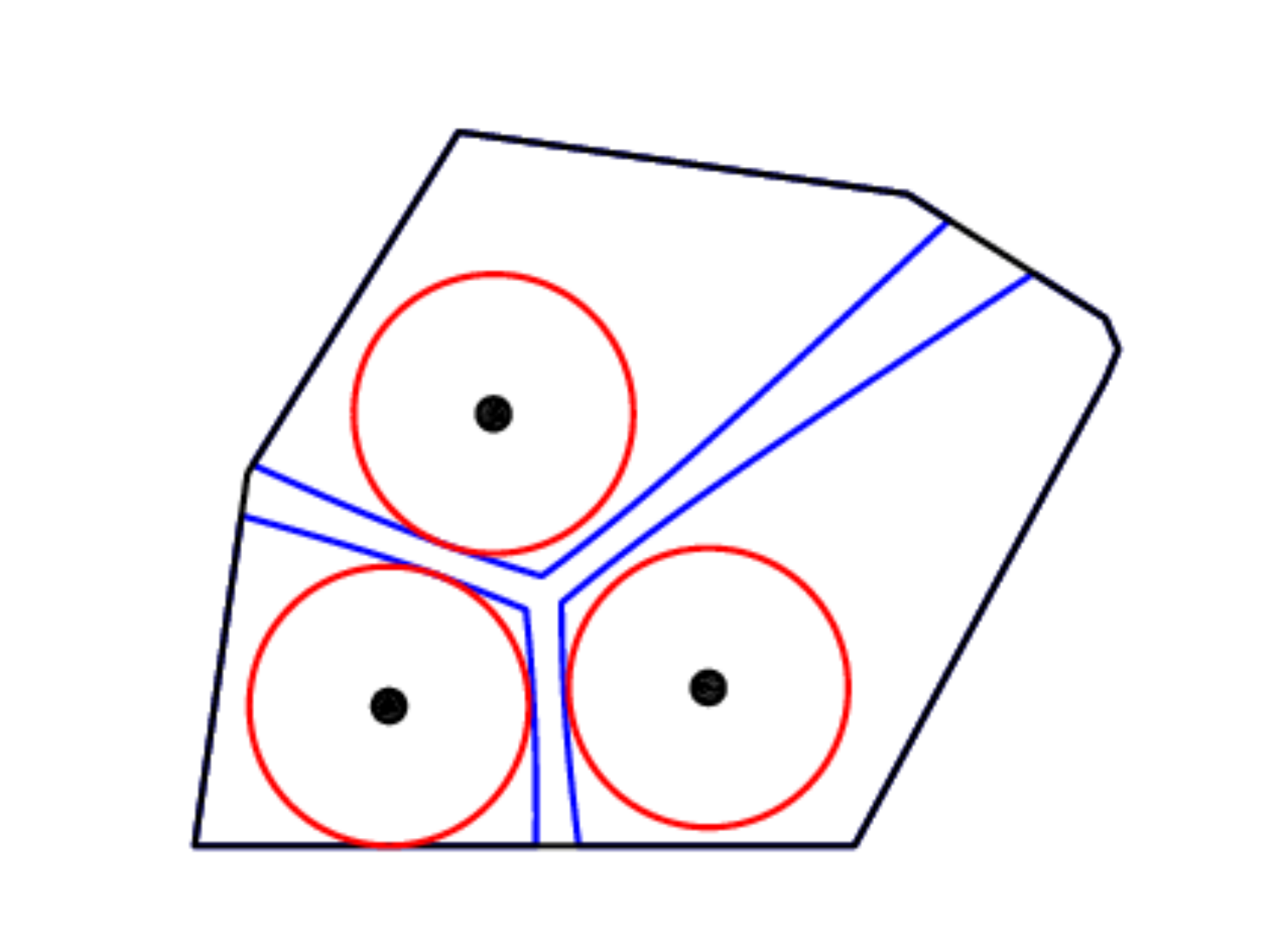} }
	\else
		\subfloat[]{ \includegraphics[width=0.6\textwidth]{figures/Case_Study_I/N3_initial.pdf} }\\
		\subfloat[]{ \includegraphics[width=0.45\textwidth]{figures/Case_Study_I/N3_final_complete.pdf} }
		\subfloat[]{ \includegraphics[width=0.45\textwidth]{figures/Case_Study_I/N3_final_simplified.pdf} }
	\fi
	
	\caption{Case Study I: (a) Initial state, (b) Final state using the optimal control law, (c) Final state using the suboptimal control law.}
	\label{fig:sim1_network}
\end{figure}

\begin{figure}[htbp]
	\centering
	\ifx\singlecol\undefined
		\includegraphics[width=0.45\textwidth]{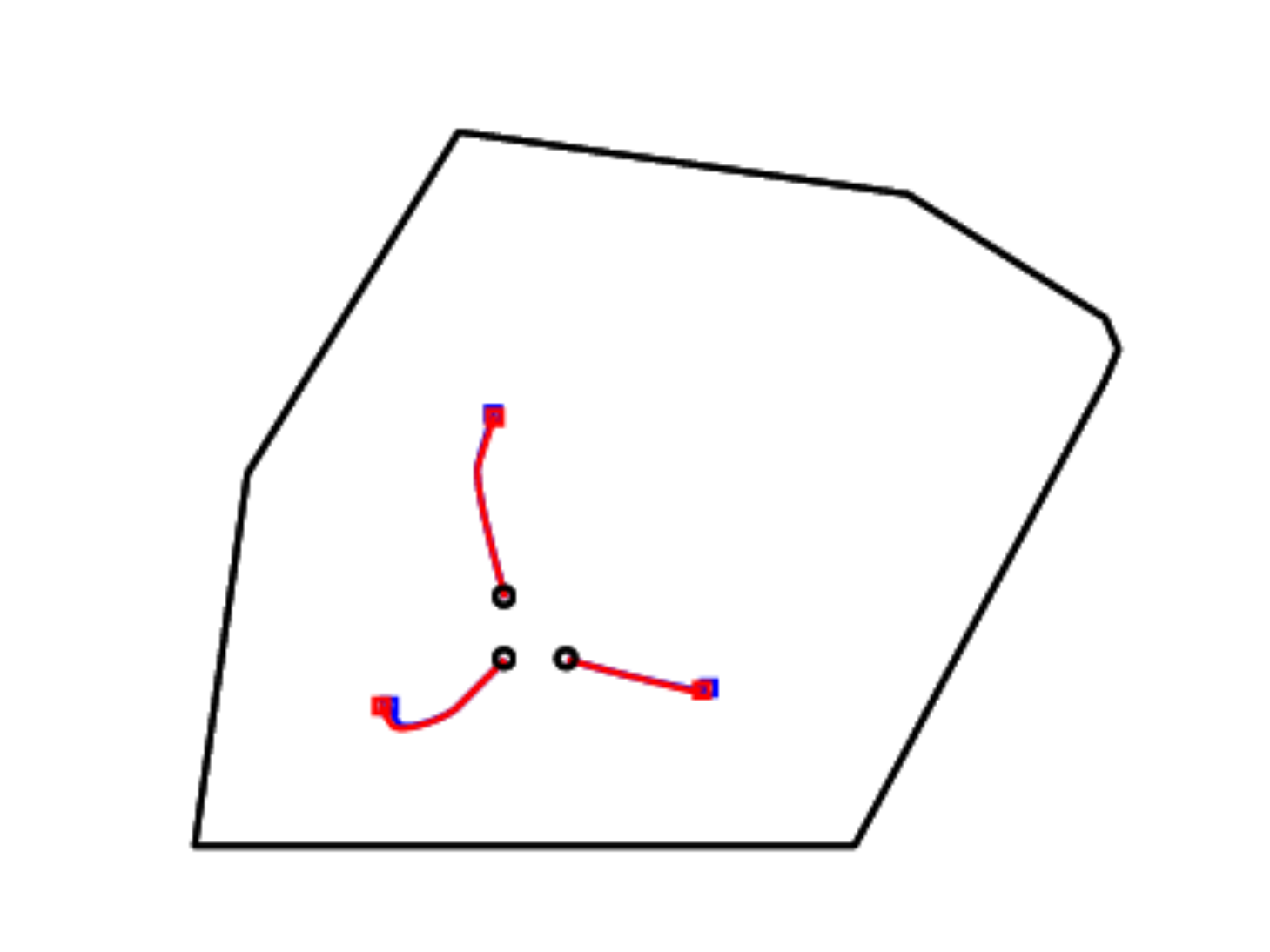}
	\else
		\includegraphics[width=0.9\textwidth]{figures/Case_Study_I/N3_traj_comp.pdf}
	\fi
	\caption{The trajectories of the nodes using the optimal (red) and suboptimal (blue) control laws.}
	\label{fig:sim1_traj}
\end{figure}

\begin{figure}[htbp]
	\centering
	\ifx\singlecol\undefined
		\includegraphics[width=0.45\textwidth]{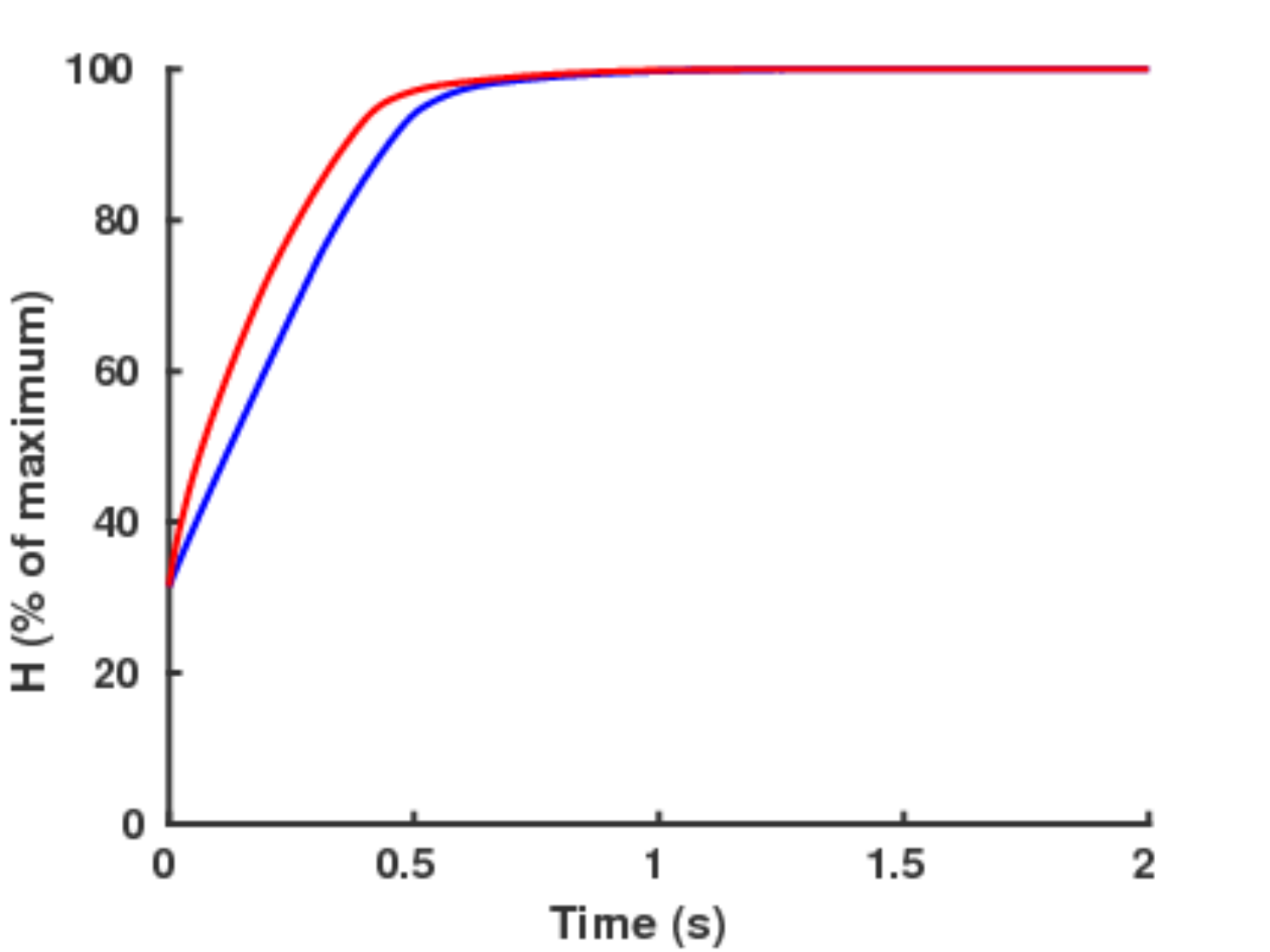}
	\else
		\includegraphics[width=0.9\textwidth]{figures/Case_Study_I/N3_H_comp.pdf}
	\fi
	\caption{The coverage objective $\mathcal{H}$ using the optimal (red) and suboptimal (blue) control laws.}
	\label{fig:sim1_H}
\end{figure}

\subsection{Case Study II}
In this case study network of ten nodes is examined. Figure \ref{fig:sim2_network} shows the initial and final states of the network, with the uncertainty disks shown in black, the guaranteed sensing disks in red and the GV cells in blue. In Figure \ref{fig:sim2_traj} the trajectories of the nodes when using the optimal control law are shown in red and when using the suboptimal control law in blue, while the initial positions are represented by circles and the final by squares. It is observed that the final configurations in both cases seem very similar at first glance, however the average distance between nodes' final states is $10.12 \%$ of the region's diameter. Careful examination of the trajectories shows that not all nodes starting from the same initial position converge to the same final position, explaining the large average distance. . Figure \ref{fig:sim2_H} shows the evolution of the coverage objective $\mathcal{H}$ with time, again shown in red for the optimal control law and in blue for the suboptimal. As it is expected, the optimal law converges faster than the suboptimal. The final value of the objective function is the same for both control laws in this case because of the similarity of their final configurations. However, since both control laws converge to local maxima, there are no guarantees that one or the other will always achieve better performance.

This case study highlights the need for the control enhancement in Section \ref{sec:constrain_movement}, especially in the case of the optimal control law. It can be seen in Figure \ref{fig:sim2_traj} that a node  moving under the optimal law (red trajectory) was stopped from moving outside the region on the lower left corner. 

\begin{figure}[htbp]
	\centering
	\ifx\singlecol\undefined
		\subfloat[]{ \includegraphics[width=0.3\textwidth]{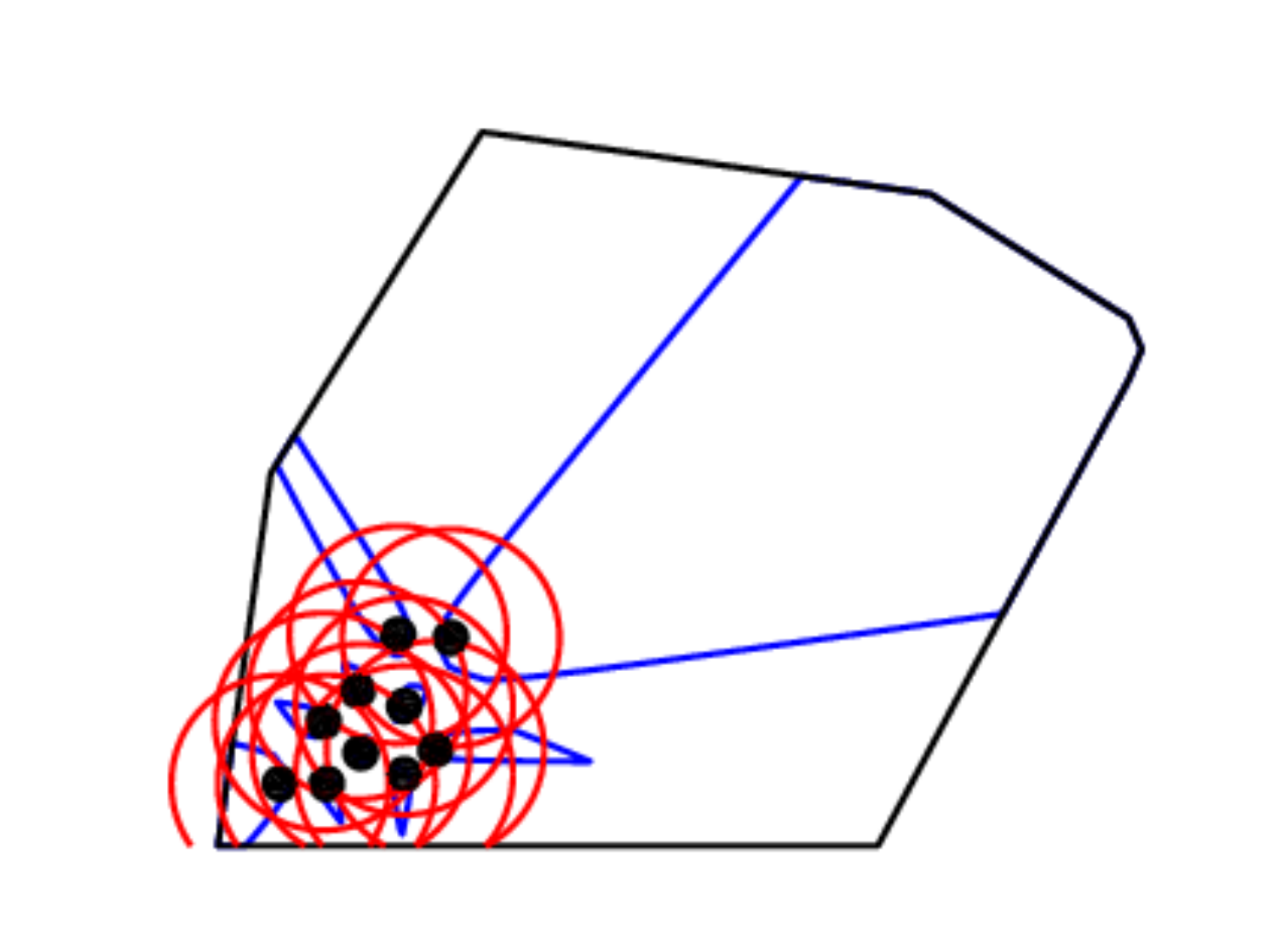} }\\
		\subfloat[]{ \includegraphics[width=0.23\textwidth]{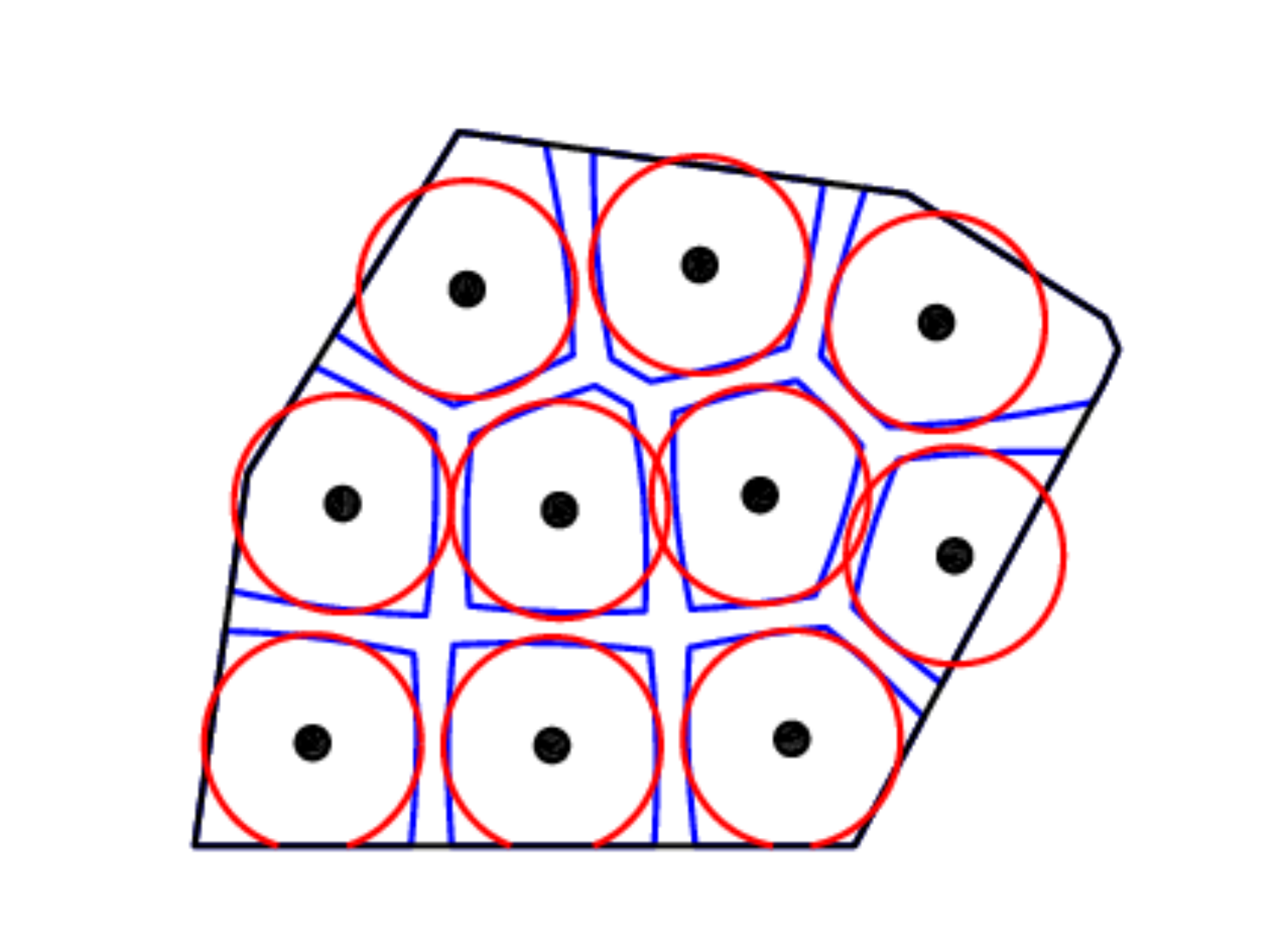} }
		\subfloat[]{ \includegraphics[width=0.23\textwidth]{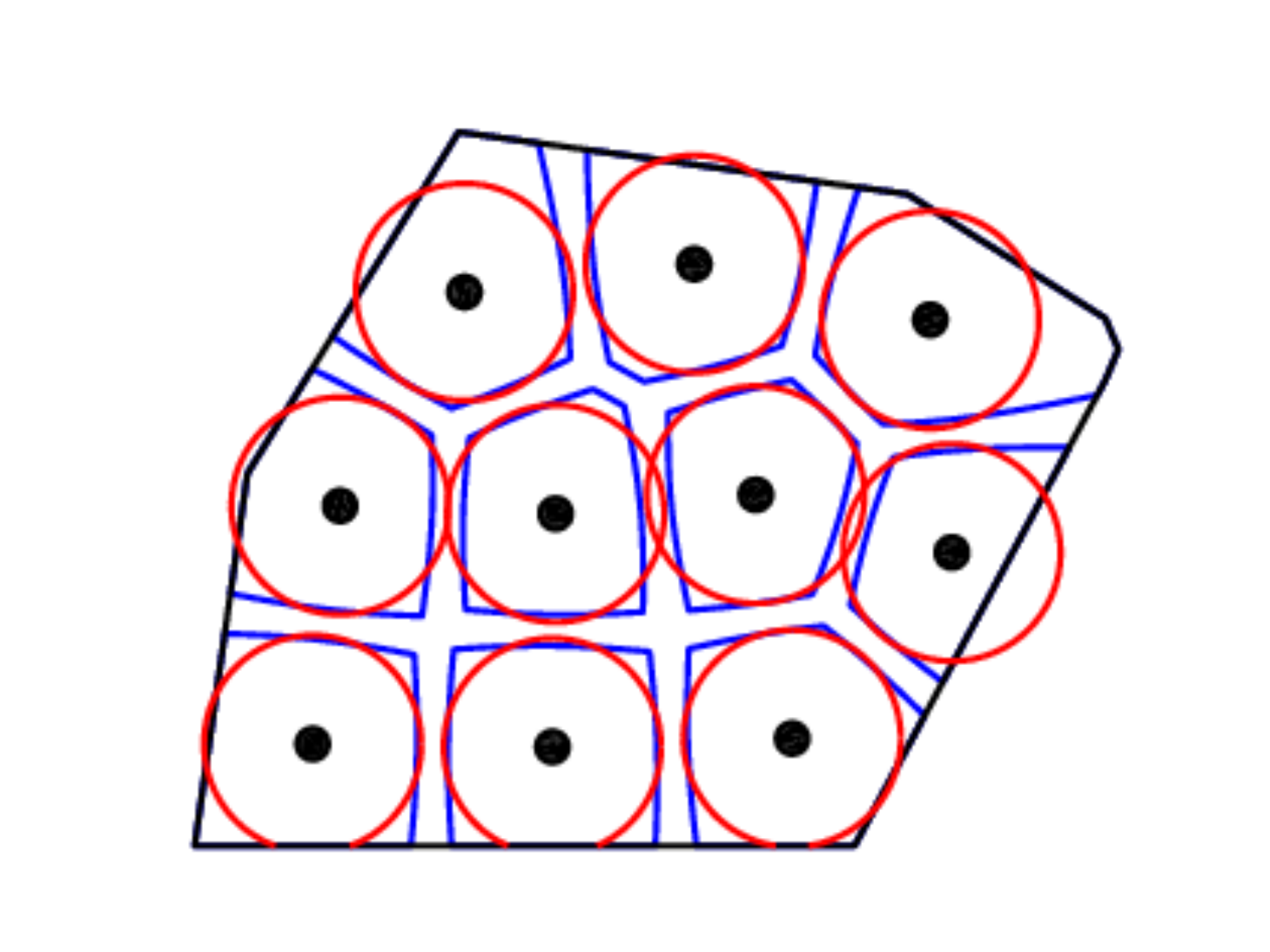} }
	\else
		\subfloat[]{ \includegraphics[width=0.6\textwidth]{figures/Case_Study_II/N10_initial.pdf} }\\
		\subfloat[]{ \includegraphics[width=0.45\textwidth]{figures/Case_Study_II/N10_final_complete.pdf} }
		\subfloat[]{ \includegraphics[width=0.45\textwidth]{figures/Case_Study_II/N10_final_simplified.pdf} }
	\fi
	
	\caption{Case Study I: (a) Initial state, (b) Final state using the optimal control law, (c) Final state using the suboptimal control law.}
	\label{fig:sim2_network}
\end{figure}

\begin{figure}[htbp]
	\centering
	\ifx\singlecol\undefined
		\includegraphics[width=0.45\textwidth]{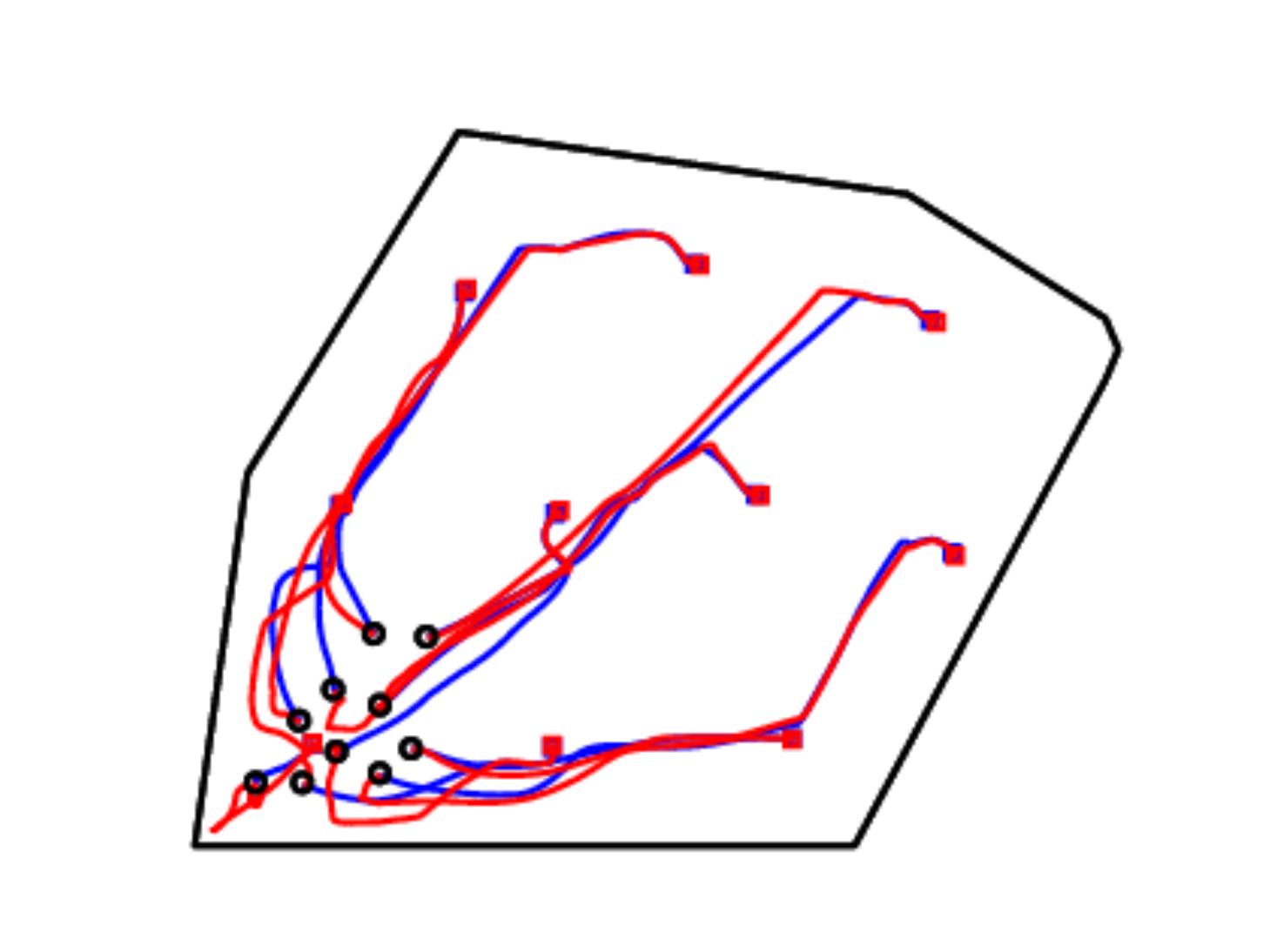}
	\else
		\includegraphics[width=0.9\textwidth]{figures/Case_Study_II/N10_traj_comp.pdf}
	\fi
	
	\caption{The trajectories of the nodes using the optimal (red) and suboptimal (blue) control laws.}
	\label{fig:sim2_traj}
\end{figure}

\begin{figure}[htbp]
	\centering
	\ifx\singlecol\undefined
		\includegraphics[width=0.45\textwidth]{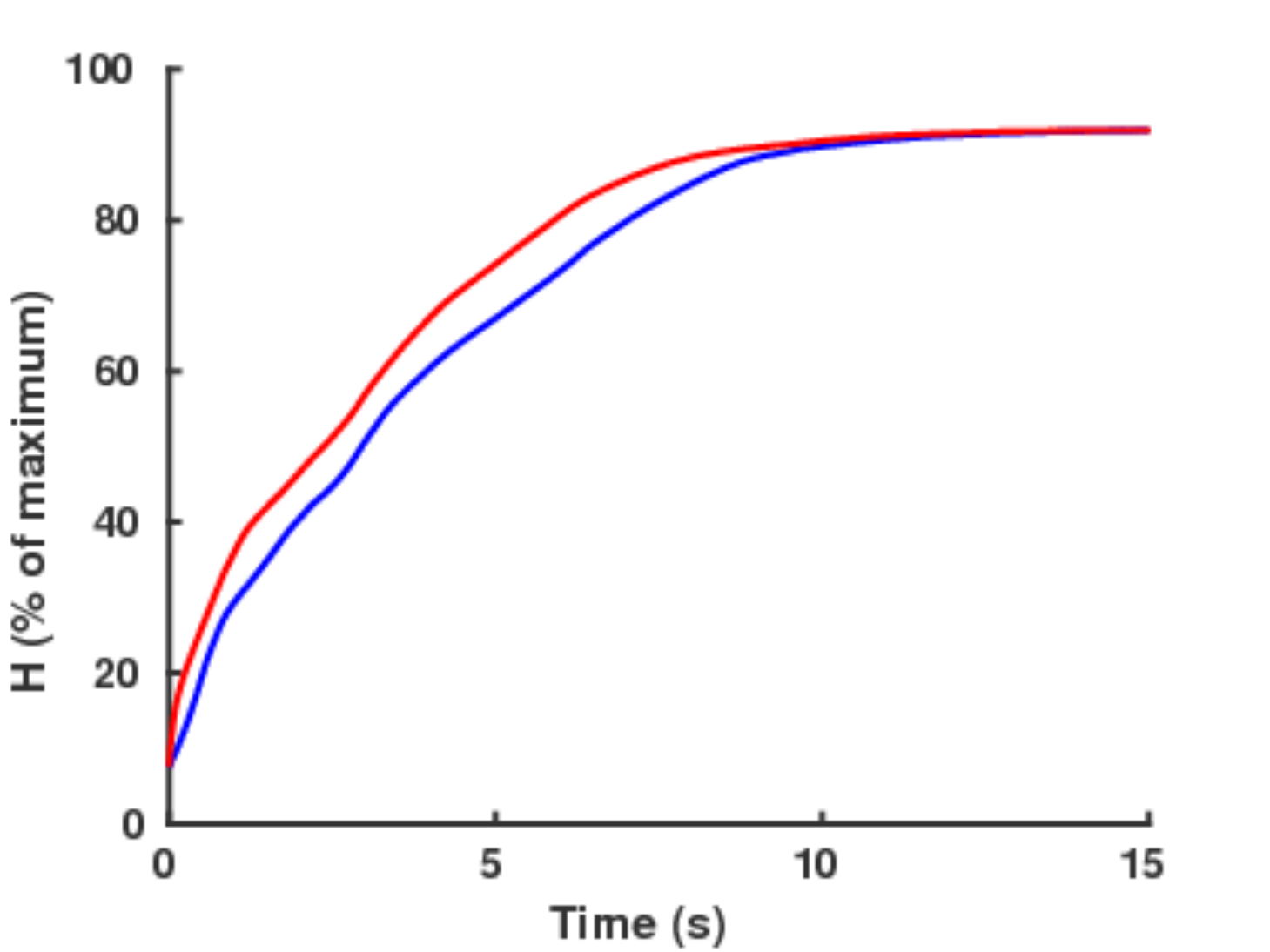}
	\else
		\includegraphics[width=0.9\textwidth]{figures/Case_Study_II/N10_H_comp.pdf}
	\fi
	
	\caption{The coverage objective $\mathcal{H}$ using the optimal (red) and suboptimal (blue) control laws.}
	\label{fig:sim2_H}
\end{figure}


\section{Experimental Studies}
\label{sec:experiments}
In order to evaluate the simplified control law (\ref{eq:control_suboptimal}), two experiments were planned and conducted. The experiments consist of $3$ robots, an external pose tracking system and a computer, communicating with the robots and the pose tracking system via a WiFi router.

The robots used are the differential-drive \emph{AmigoBot}s by \emph{ActiveMedia Robotics}, which provide a high level command set, such as setting their translational and rotational velocities directly. These robots are also equipped with encoders, used for self-estimating their pose, but due to the drifting error, an external tracking system was implemented. Although the external pose estimation system is used for the control law implementation, its performance is compared with that of the encoders. Since the robots are not equipped with actual sensors, circular sensing patters with a radius $r^s = 0.3~m$ were assumed in the control law implementation.

The pose estimation system consists of a camera and an ODROID-XU4 octa core computer, running an ArUco based pose estimation server. ArUco~\cite{Aruco2014} is a minimal library for Augmented Reality applications and provides pose estimation of some predefined sets of fiducial markers. Hende, by attaching a fiducial marker on top of each robot, the ArUco library provides an estimation of each robot's pose.

In order to estimate the positioning uncertainty of the experimental setup, a fiducial marker was placed on the centroid and on each of the vertices of the region $\Omega$. A 3D triangular mesh was created from these measurements as shown in Figure \ref{fig:barycentric}. The maximum uncertainty value from that mesh $r^{u_{\max}} = 0.032$~m was used for the experiments. 

\begin{figure}[htb]
	\centering
	\ifx\singlecol\undefined
		\includegraphics[width=0.5\textwidth]{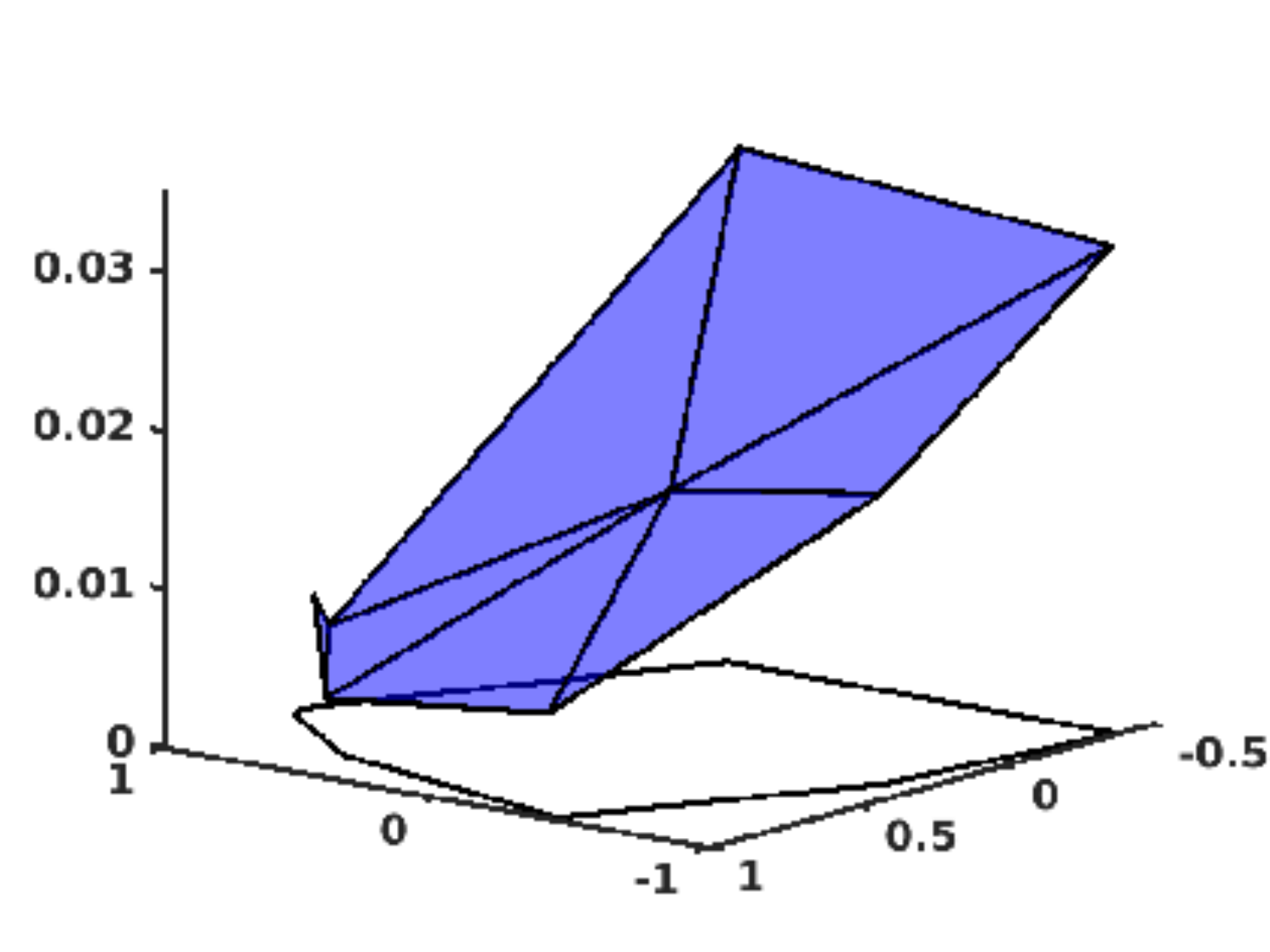}
	\else
		\includegraphics[width=0.95\textwidth]{figures/barycentric.pdf}
	\fi
	
	\caption{Triangular mesh from positioning uncertainty measurements.}
	\label{fig:barycentric}
\end{figure}


The control law was implemented algorithmically as a loop with a period $T_s = 0.1$~sec. After the end of each iteration the computer sends velocity commands to the robots. At first, the robots' current positions are used to create the Guaranteed Voronoi diagram of the uncertain disks and calculate the control inputs according to the simplified control law~(\ref{eq:control_suboptimal}). Subsequently, a target point $q_i^t \in \mathbb{R}^2$ is created for each robot based on its current position $q_i$ and its control input $\tilde{u_i}$. Once the target point for each robot has been found, and given each robot's position $q_i$ and orientation $\theta_i$, the translational $v_i$ and rotational $\omega_i$ velocity control inputs are sent to each robot.
\begin{align*}
v_i &= \min \left( \frac{\| q_i^t - q_i \|}{T_s}, v_{\max} \right) \cos(d\theta_i) \\
\omega &= \min \left( \frac{ | d\theta_i | }{T_s}, \omega_{\max} \right) \sin(d\theta_i)
\end{align*}
where $d\theta_i = \angle(q_i^t - q_i) - \theta_i$ and $v_{\max}$, $\omega_{\max}$ are constraints on the maximum translational and rotational velocity respectively of the Amigobot robots. Once all robots are within a predefined distance $d_t = 0.02~m$ of their respective target points, the Guaranteed Voronoi diagram and the control law are calculated again and new target points are produced.

The region shown in all figures regarding the experiments is $\Omega^s \subseteq \Omega$ i.e. the region inside which the centers of the uncertain disks are constrained, as explained in Section \ref{sec:constrain_movement}. 

Both experiments are evaluated against simulations with the same initial robot positions, positioning uncertainty and sensing performance.

\subsection{Experiment I}
The initial and final positions of the robots for both the experiment and the simulation are shown in Figures \ref{fig:exp1_initial} and \ref{fig:exp1_final} respectively with the uncertainty disks shown in black, the guaranteed sensing disks in red and the GV cells in blue. Additionally, photos of the experiment initial and final configurations can be seen on Figure \ref{fig:exp1_photos}. It can be seen that in the final configuration all guaranteed sensing disks are contained within their respective GV cells in both the experiment and the simulation, thus leading the network to a globally optimal configuration. By comparing the final configurations, as well as the robot trajectories for the experiment and simulation shown in Figure \ref{fig:exp1_traj_comp}, it can be seen that they differ significantly. The mean distance between the nodes final positions in the experiment and simulation is $8.68 \%$ of the diameter of $\Omega^s$. Because in this current case there are multiple globally optimal configurations the network can possibly reach and since the implemented control law differs from the theoretical one, it is natural that the final configurations between the experiment and the simulation differ. Nevertheless, a globally optimal configuration was reached in both cases. The coverage objective $\mathcal{H}$ did not increase monotonously due to the implemented control law. It did increase however from $83.7 \%$ to $99.9 \%$ of its maximum possible value. The positioning data from the ArUco library and the robot encoders are compared in Figure \ref{fig:exp1_traj} [Left]. As it is expected, the further a robot moves and the more it rotates, the larger the positioning error of the encoders grows. The mean error between the ArUco and encoder final robot positions is $4.04 \%$ of the diameter of $\Omega^s$. Figure \ref{fig:exp1_traj} [Right] shows the trajectories of the target points used in the control law implementation. 

\begin{figure}[htbp]
	\centering
	\ifx\singlecol\undefined
		\includegraphics[width=0.24\textwidth]{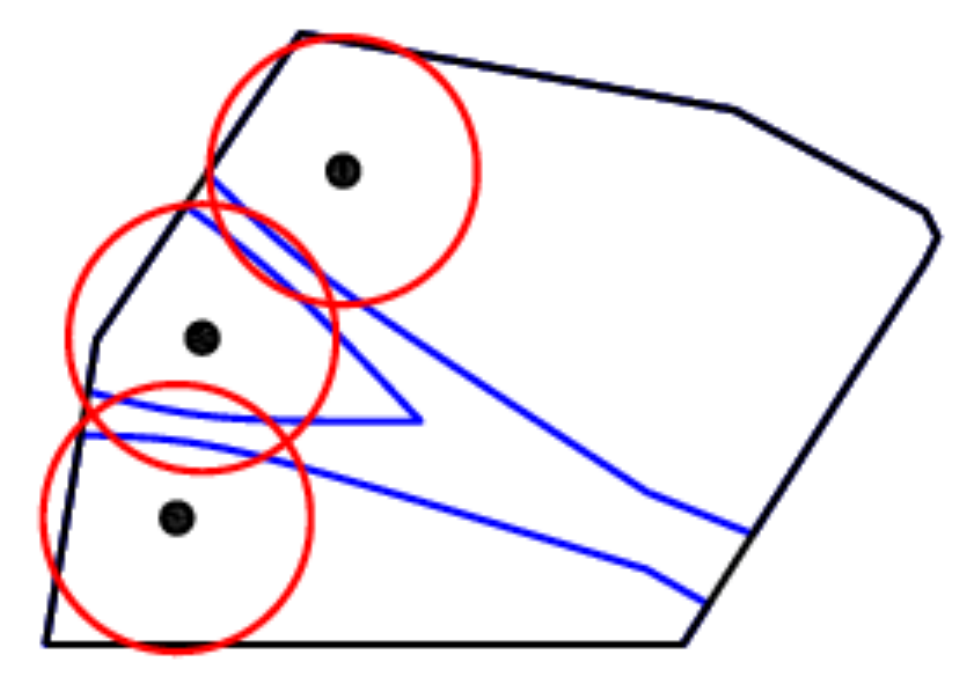}
	\else
		\includegraphics[width=0.5\textwidth]{figures/Experiment_I/exp1_initial.pdf}
	\fi
	
	\caption{Experiment I: Initial configuration.}
	\label{fig:exp1_initial}
\end{figure}

\begin{figure}[htbp]
	\centering
	\ifx\singlecol\undefined
		\includegraphics[width=0.24\textwidth]{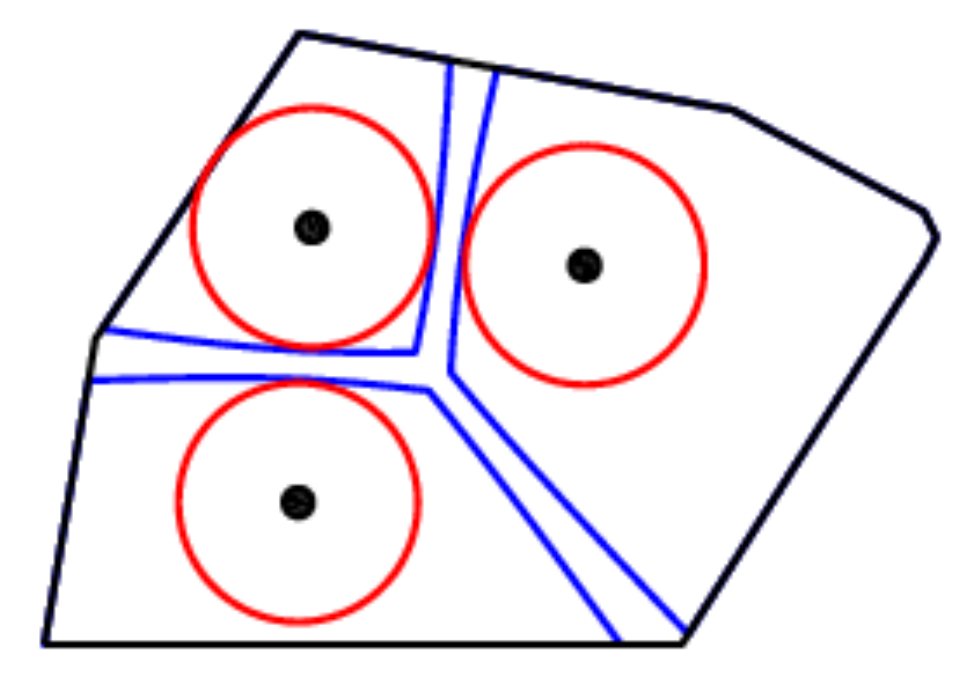}
		\includegraphics[width=0.24\textwidth]{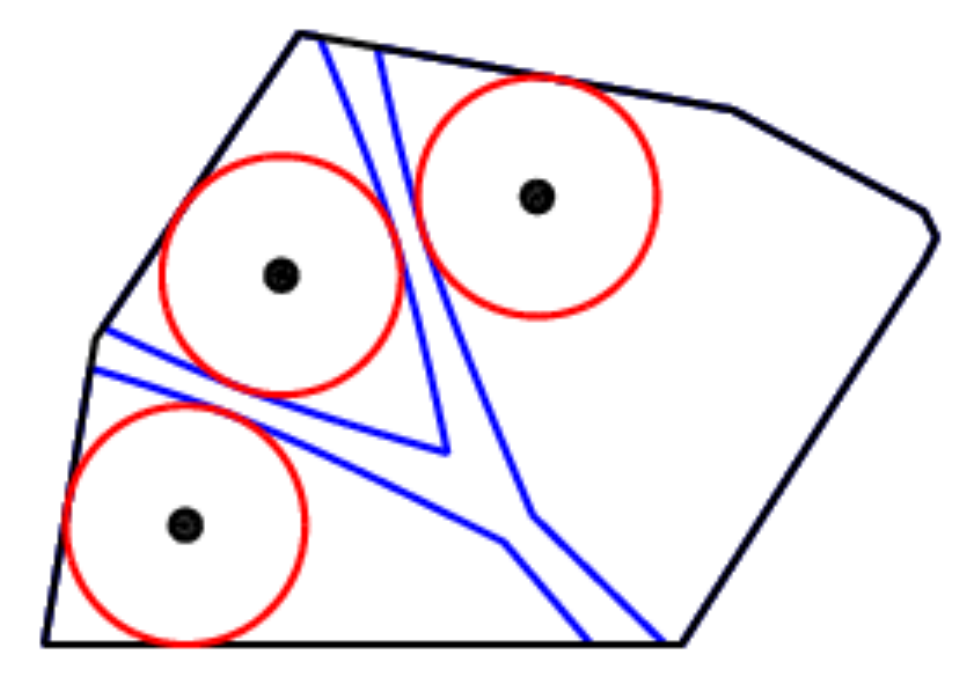}
	\else
		\includegraphics[width=0.48\textwidth]{figures/Experiment_I/exp1_final.pdf}
		\includegraphics[width=0.48\textwidth]{figures/Experiment_I/exp1_final_sim.pdf}
	\fi
	
	\caption{Experiment I: Experiment [Left] and simulation [Right] final configuration.}
	\label{fig:exp1_final}
\end{figure}

\begin{figure}[htbp]
	\centering
	\ifx\singlecol\undefined
		\includegraphics[width=0.24\textwidth]{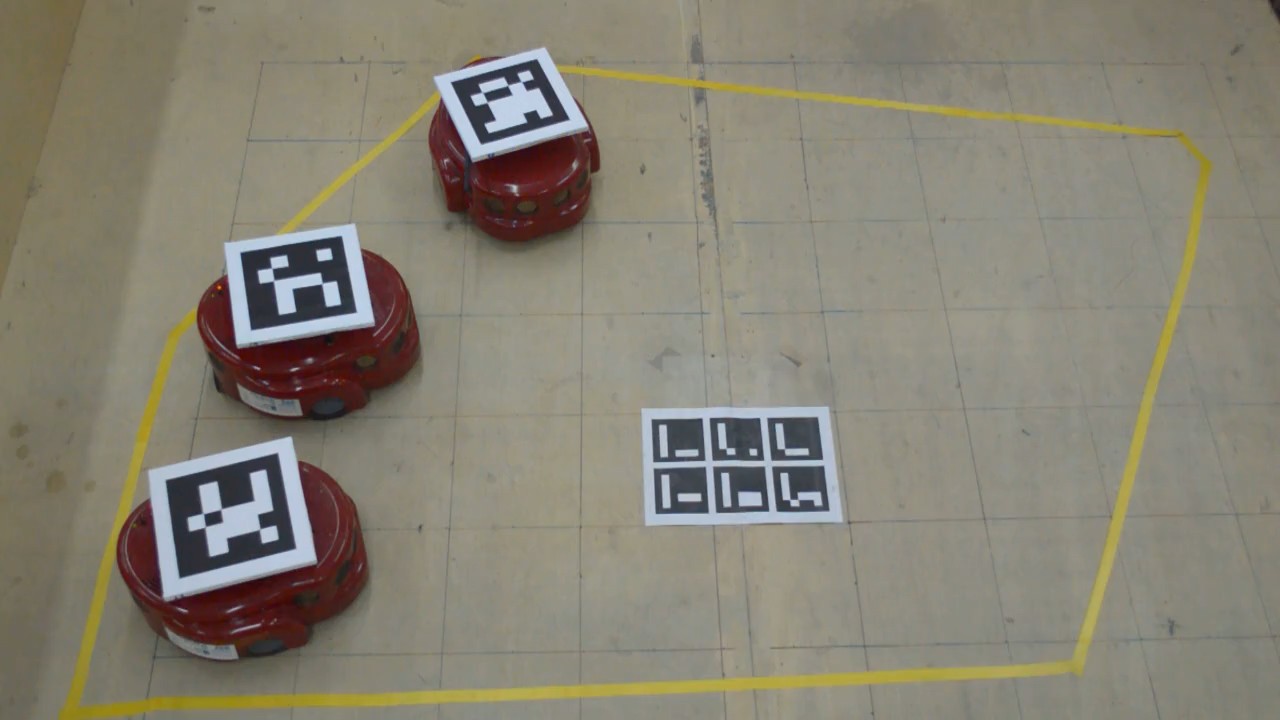}
		\includegraphics[width=0.24\textwidth]{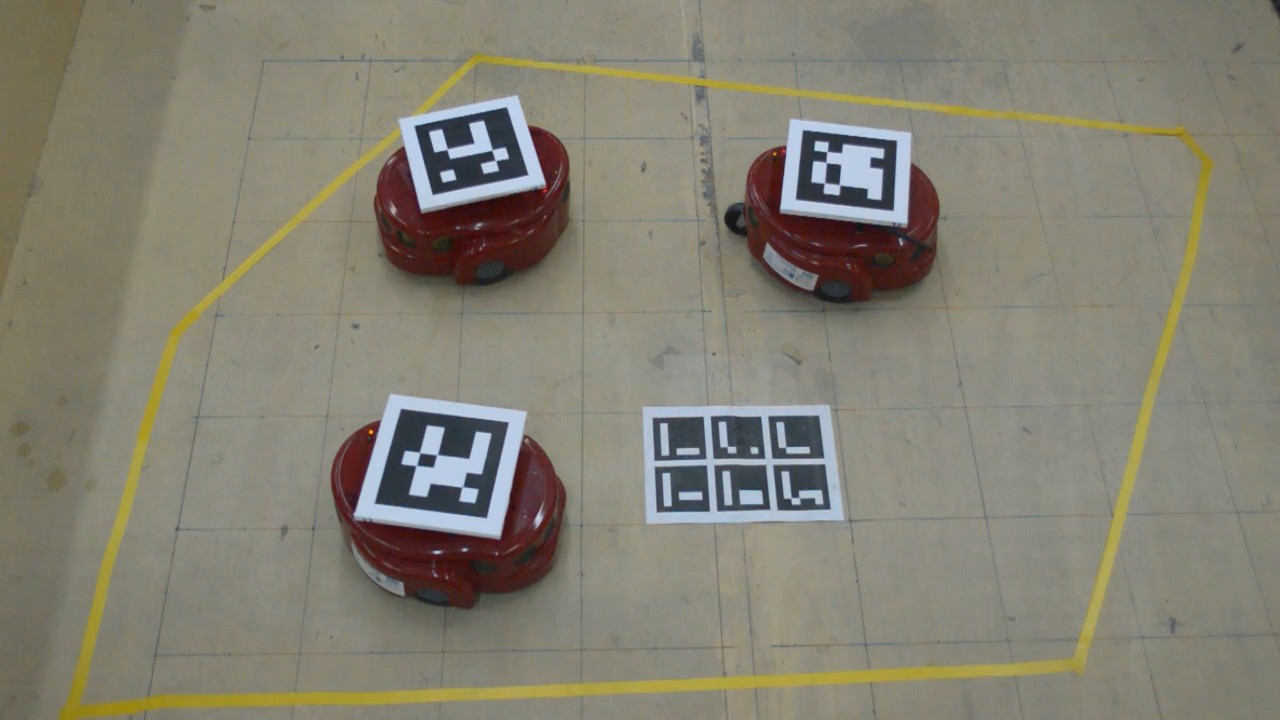}
	\else
		\includegraphics[width=0.48\textwidth]{figures/Experiment_I/exp1_initial.jpg}
		\includegraphics[width=0.48\textwidth]{figures/Experiment_I/exp1_final.jpg}
	\fi
	
	\caption{Experiment I: Photos from the initial [Left] and final [Right] robot positions in the experiment.}
	\label{fig:exp1_photos}
\end{figure}

\begin{figure}[htbp]
	\centering
	\ifx\singlecol\undefined
		\includegraphics[width=0.3\textwidth]{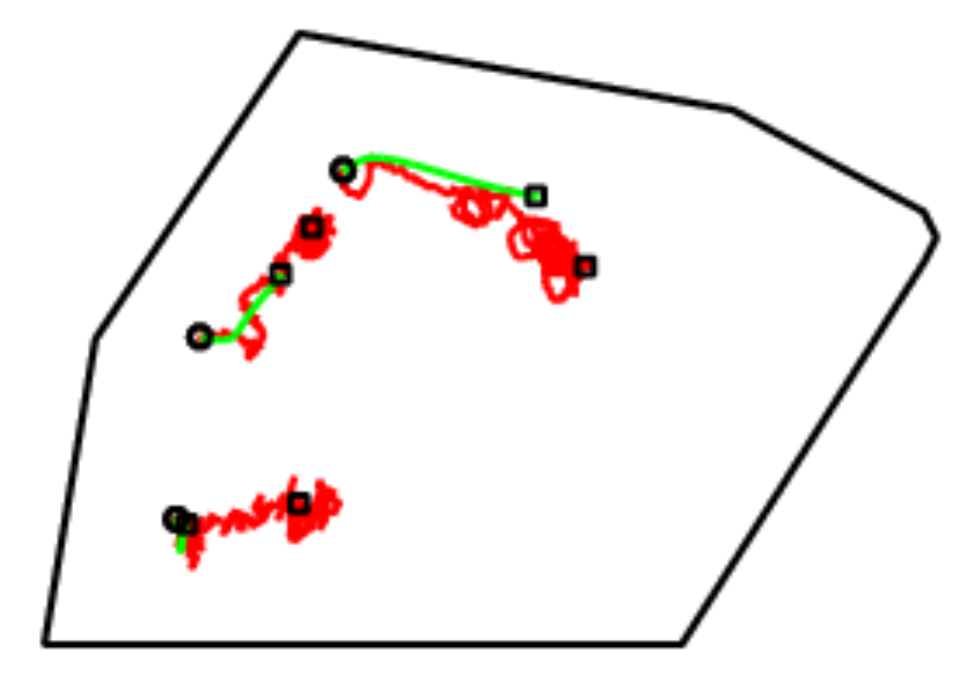}
	\else
		\includegraphics[width=0.6\textwidth]{figures/Experiment_I/exp1_traj_comp_sim.pdf}
	\fi
	
	\caption{Experiment I: Comparison of experiment (red) and simulation (green) robot trajectories. Initial positions marked with circles and final positions with squares.}
	\label{fig:exp1_traj_comp}
\end{figure}

\begin{figure}[htbp]
	\centering
	\ifx\singlecol\undefined
		\includegraphics[width=0.45\textwidth]{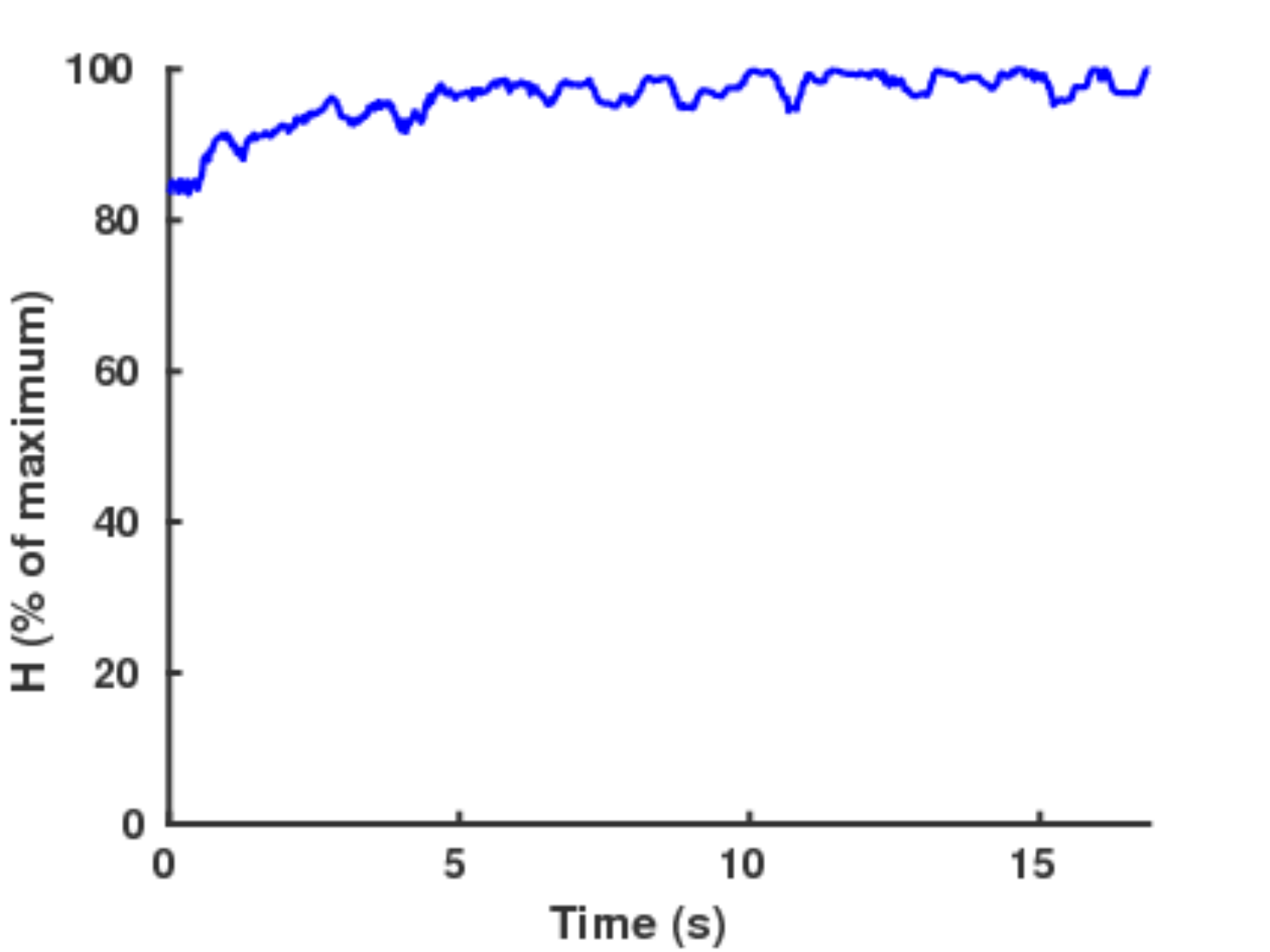}
	\else
		\includegraphics[width=0.9\textwidth]{figures/Experiment_I/exp1_H.pdf}
	\fi
	
	\caption{Experiment I: Coverage objective $\mathcal{H}$ as a function of time.}
	\label{fig:exp1_H}
\end{figure}

\begin{figure}[htbp]
	\centering
	\ifx\singlecol\undefined
		\includegraphics[width=0.24\textwidth]{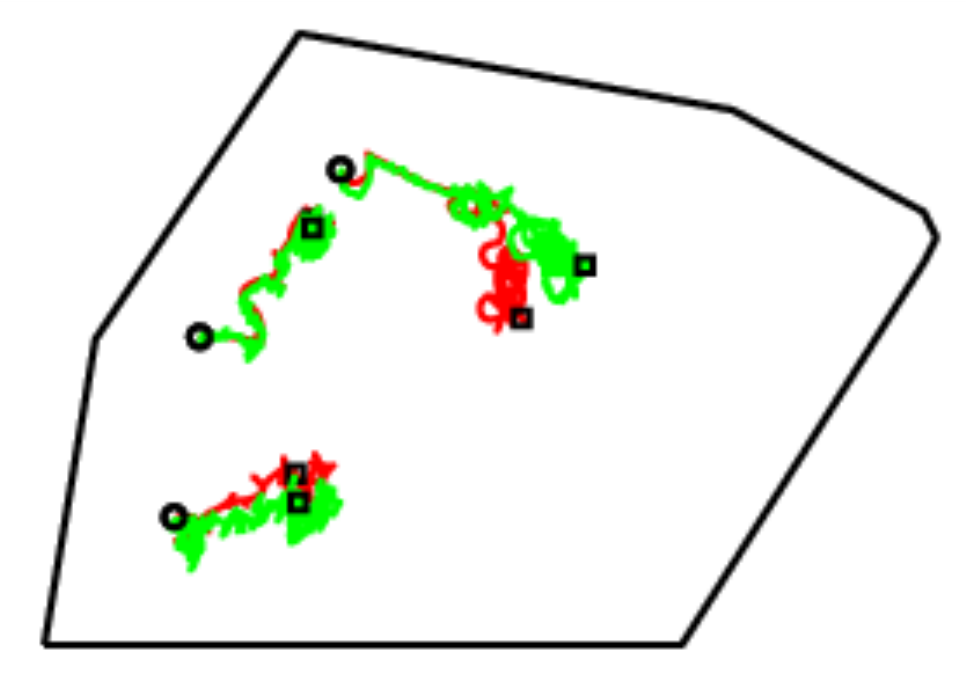}
		\includegraphics[width=0.24\textwidth]{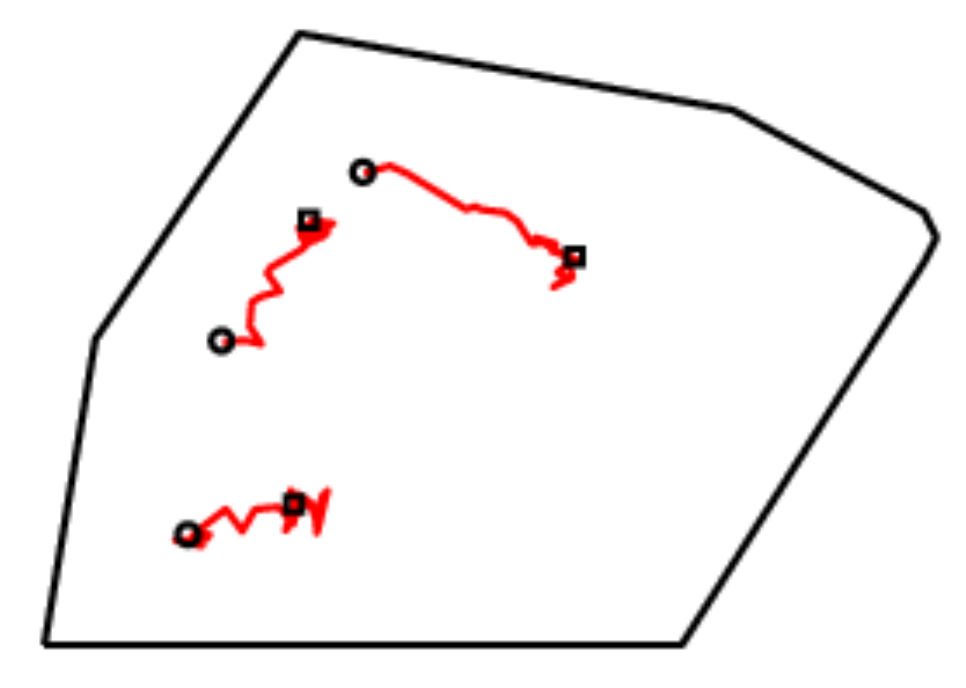}
	\else
		\includegraphics[width=0.48\textwidth]{figures/Experiment_I/exp1_traj.pdf}
		\includegraphics[width=0.48\textwidth]{figures/Experiment_I/exp1_target.pdf}
	\fi
	
	\caption{Experiment I: [Left] ArUco (green) and encoder (red) reported trajectories. [Right] Target point trajectories. Initial positions marked with circles and final positions with squares.}
	\label{fig:exp1_traj}
\end{figure}

\subsection{Experiment II}
In this experiment, half the diameter of the robots was added to their positioning uncertainty, so that each robot was completely contained within its positioning uncertainty disk $C_i^u$. The guaranteed sensing disks however were calculated using only the actual uncertainty and not the one increased by the robot radius. This was done to guarantee that the robots would remain inside the region $\Omega$. The initial and final positions of the robots for both the experiment and the simulation are shown in Figures \ref{fig:exp2_initial} and \ref{fig:exp2_final} respectively with the uncertainty disks shown in black, the guaranteed sensing disks in red and the GV cells in blue. Additionally, photos of the experiment initial and final configurations can be seen on Figure \ref{fig:exp2_photos}. It can be seen that in the final configuration all guaranteed sensing disks are almost completely contained within their respective GV cells in both the experiment and the simulation. By comparing the final configurations, as well as the robot trajectories for the experiment and simulation shown in Figure \ref{fig:exp2_traj_comp}, it can be seen that they do not differ significantly. The mean distance between the nodes final positions in the experiment and simulation is $1.54 \%$ of the diameter of $\Omega^s$. Because of the increased positioning uncertainty in this case, which results in smaller GV cells, the number of local maxima of the objective function is significantly smaller than in the previous experiment. Thus it is quite possible that, despite their differences, the theoretical and implemented control laws converge to the same final configuration. The coverage objective $\mathcal{H}$ did not increase monotonously due to the implemented control law. It did increase however from $38.8 \%$ to $98.3 \%$ of its maximum possible value. The positioning data from the ArUco library and the robot encoders are compared in Figure \ref{fig:exp2_traj} [Left]. As it is expected, the further a robot moves and the more it rotates, the larger the positioning error of the encoders grows. The mean error between the ArUco and encoder final robot positions is $1.54 \%$ of the diameter of $\Omega^s$. Figure \ref{fig:exp2_traj} [Right] shows the trajectories of the target points used in the control law implementation. 

\ifx\singlecol\undefined
A video of the simulations and the experiments can be found on \url{http://anemos.ece.upatras.gr/images/stories/videos/PTGS_IEEETAC.mp4}
\else
A video of the simulations and the experiments can be found on \\ \texttt{http://anemos.ece.upatras.gr/images/stories/videos/PTGS\_IEEETAC.mp4}
\fi

\begin{figure}[htbp]
	\centering
	\ifx\singlecol\undefined
		\includegraphics[width=0.24\textwidth]{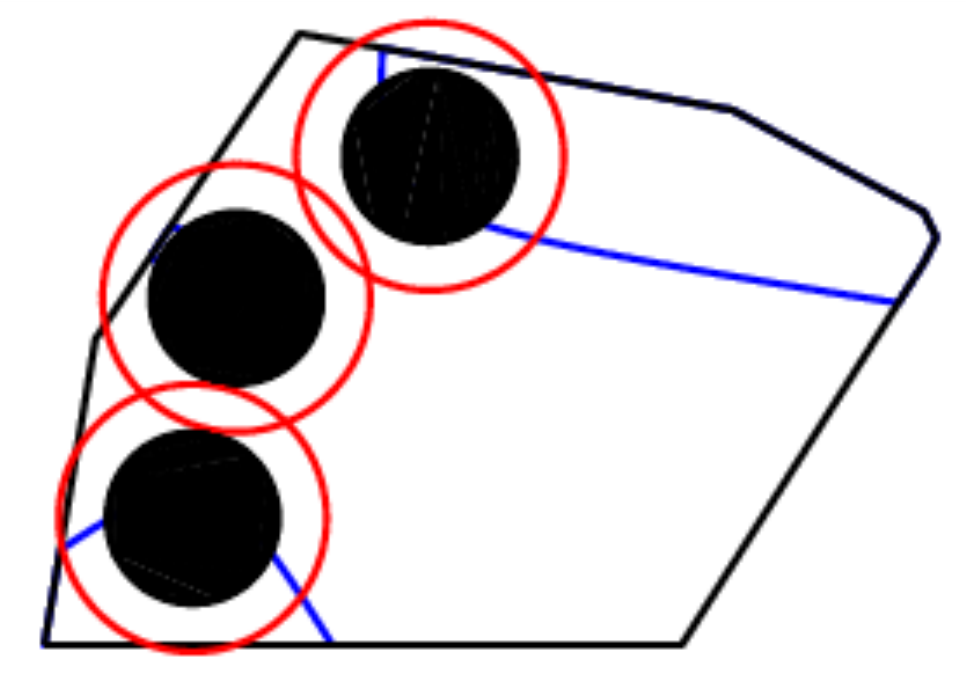}
	\else
		\includegraphics[width=0.5\textwidth]{figures/Experiment_II/exp2_initial.pdf}
	\fi
	
	\caption{Experiment II: Initial configuration.}
	\label{fig:exp2_initial}
\end{figure}

\begin{figure}[htbp]
	\centering
	\ifx\singlecol\undefined
		\includegraphics[width=0.24\textwidth]{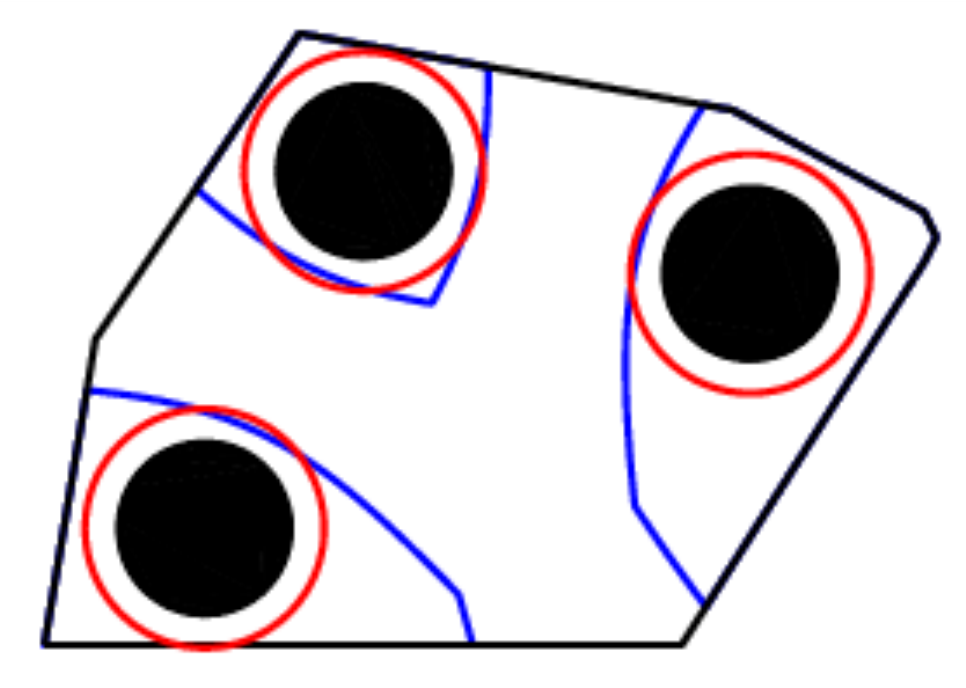}
		\includegraphics[width=0.24\textwidth]{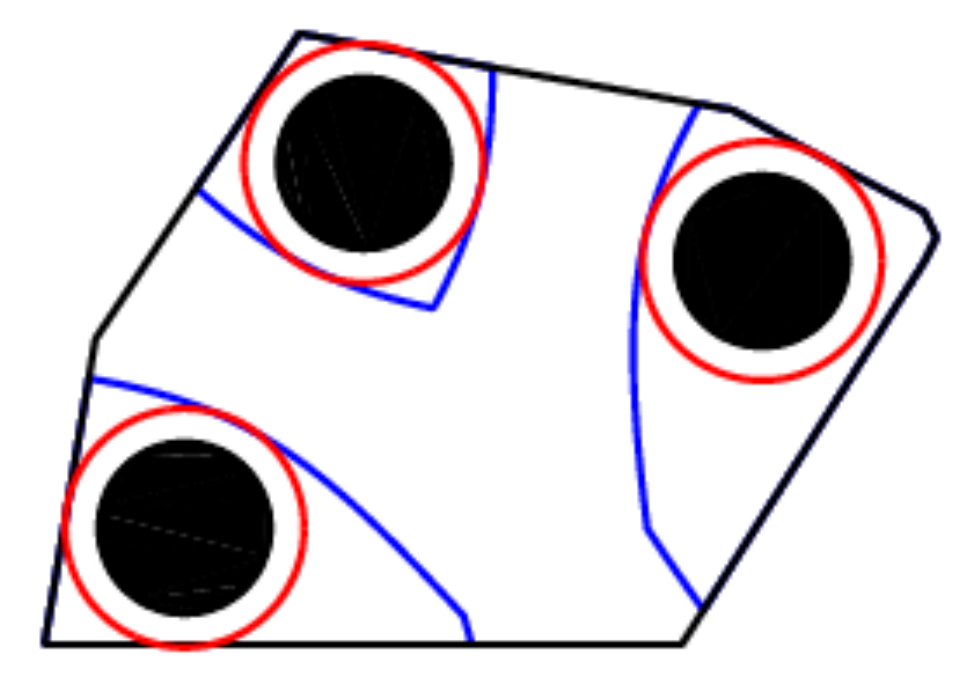}
	\else
		\includegraphics[width=0.48\textwidth]{figures/Experiment_II/exp2_final.pdf}
		\includegraphics[width=0.48\textwidth]{figures/Experiment_II/exp2_final_sim.pdf}
	\fi
	
	\caption{Experiment II: Experiment [Left] and simulation [Right] final configuration.}
	\label{fig:exp2_final}
\end{figure}

\begin{figure}[htbp]
	\centering
	\ifx\singlecol\undefined
		\includegraphics[width=0.24\textwidth]{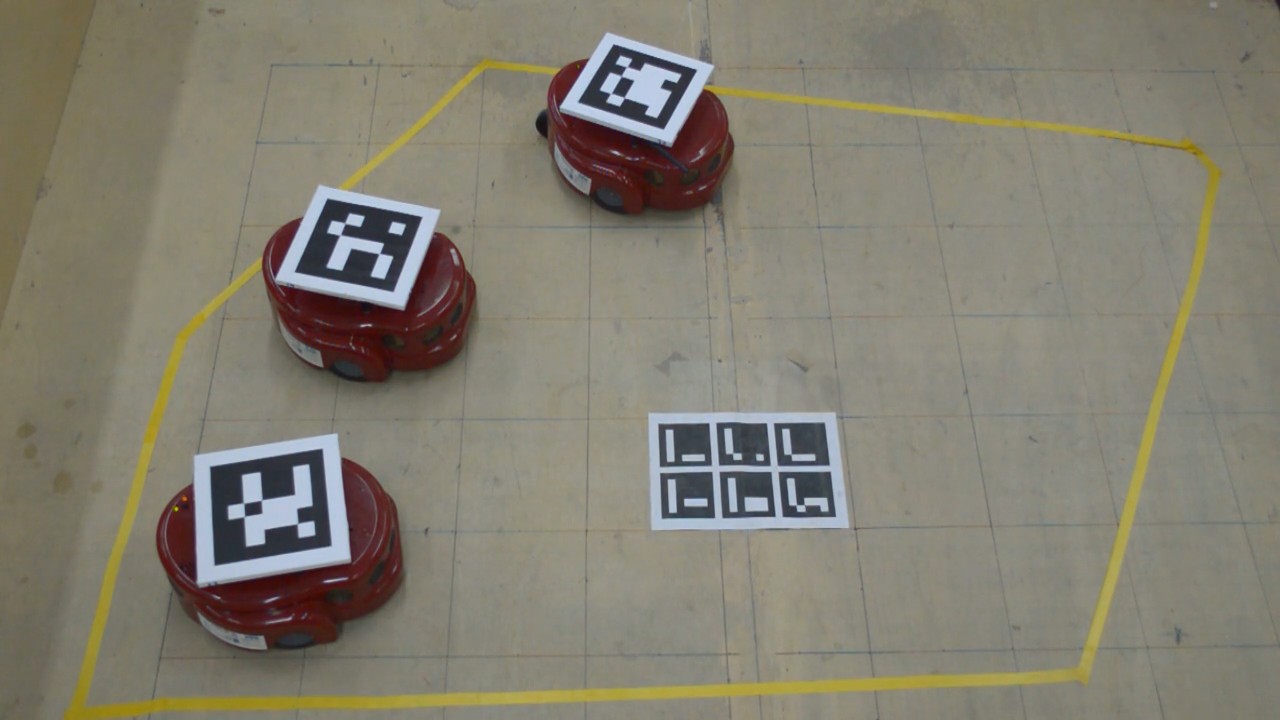}
		\includegraphics[width=0.24\textwidth]{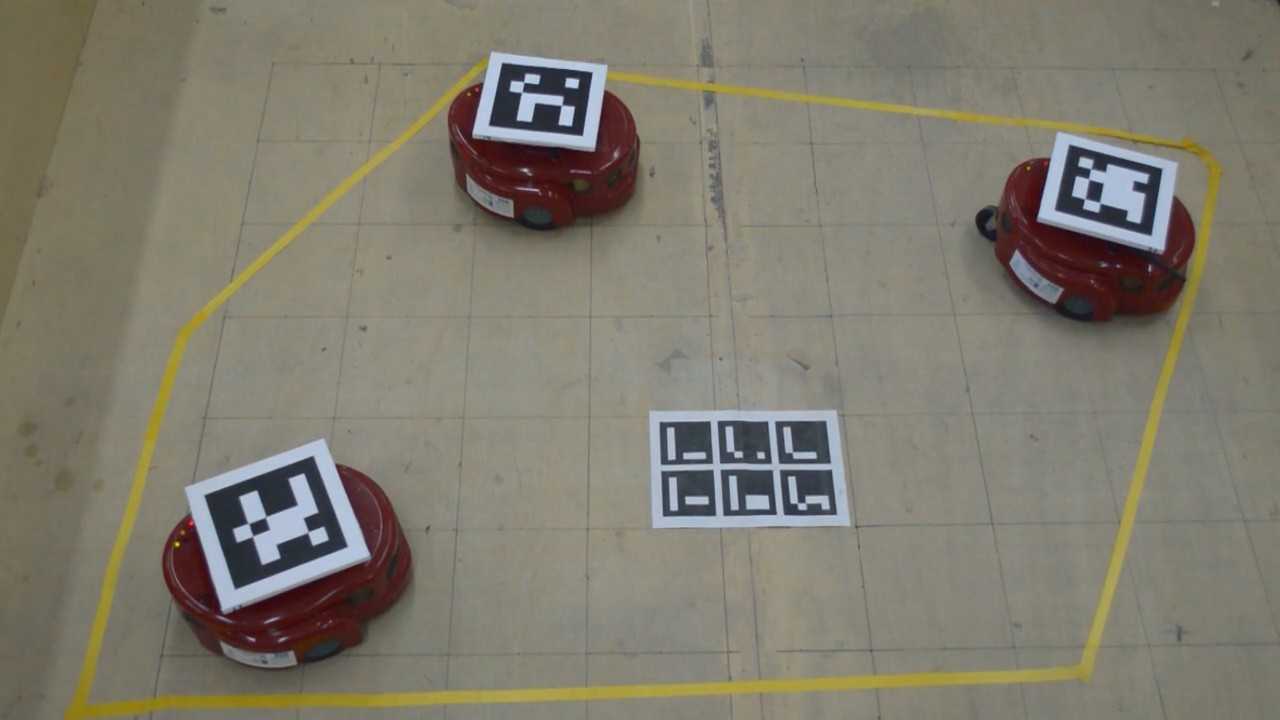}
	\else
		\includegraphics[width=0.48\textwidth]{figures/Experiment_II/exp2_initial.jpg}
		\includegraphics[width=0.48\textwidth]{figures/Experiment_II/exp2_final.jpg}
	\fi
	
	\caption{Experiment II: Photos from the initial [Left] and final [Right] robot positions in the experiment.}
	\label{fig:exp2_photos}
\end{figure}

\begin{figure}[htbp]
	\centering
	\ifx\singlecol\undefined
		\includegraphics[width=0.3\textwidth]{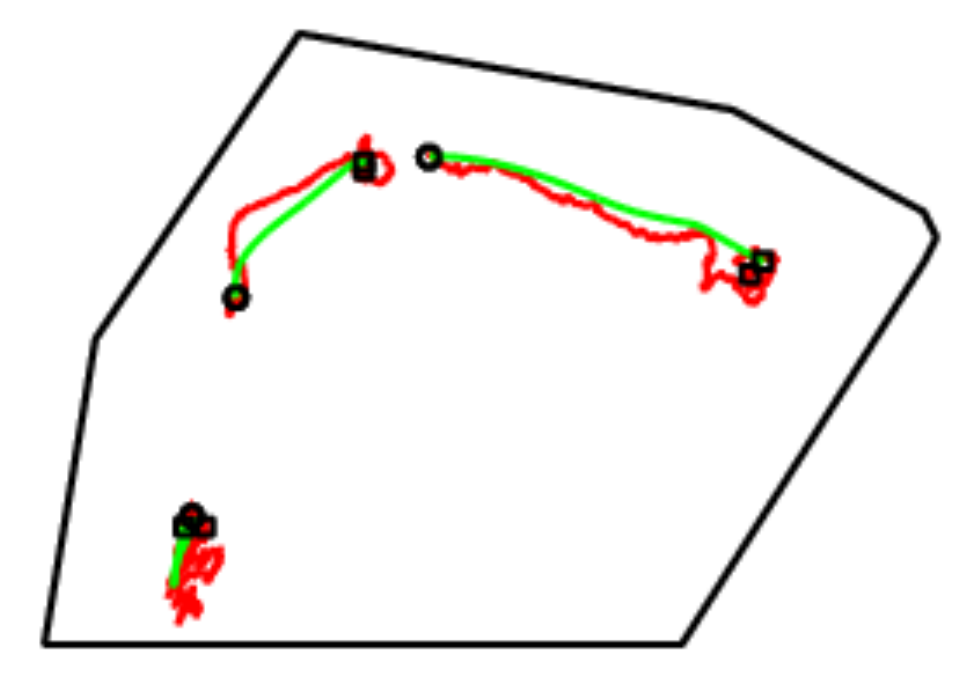}
	\else
		\includegraphics[width=0.6\textwidth]{figures/Experiment_II/exp2_traj_comp_sim.pdf}
	\fi
	
	\caption{Experiment II: Comparison of experiment (red) and simulation (green) robot trajectories. Initial positions marked with circles and final positions with squares.}
	\label{fig:exp2_traj_comp}
\end{figure}

\begin{figure}[htbp]
	\centering
	\ifx\singlecol\undefined
		\includegraphics[width=0.45\textwidth]{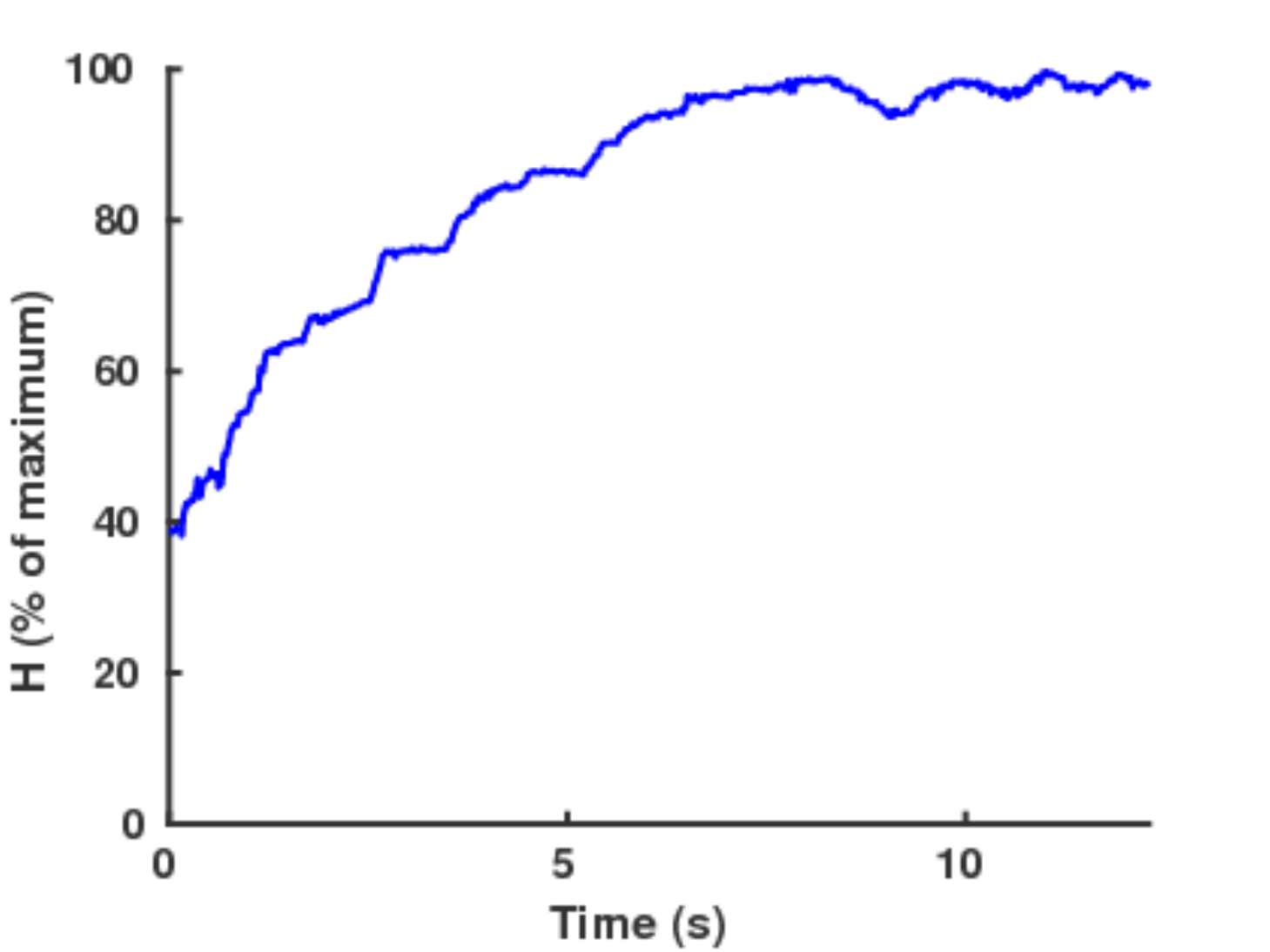}
	\else
		\includegraphics[width=0.9\textwidth]{figures/Experiment_II/exp2_H.pdf}
	\fi
	
	\caption{Experiment II: Coverage objective $\mathcal{H}$ as a function of time.}
	\label{fig:exp2_H}
\end{figure}

\begin{figure}[htbp]
	\centering
	\ifx\singlecol\undefined
		\includegraphics[width=0.24\textwidth]{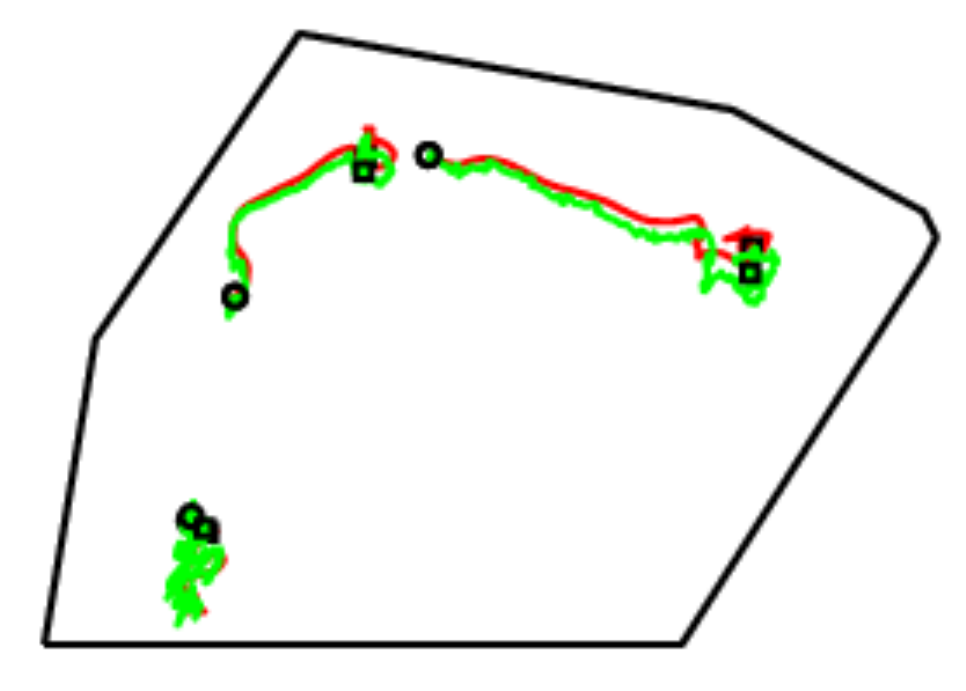}
		\includegraphics[width=0.24\textwidth]{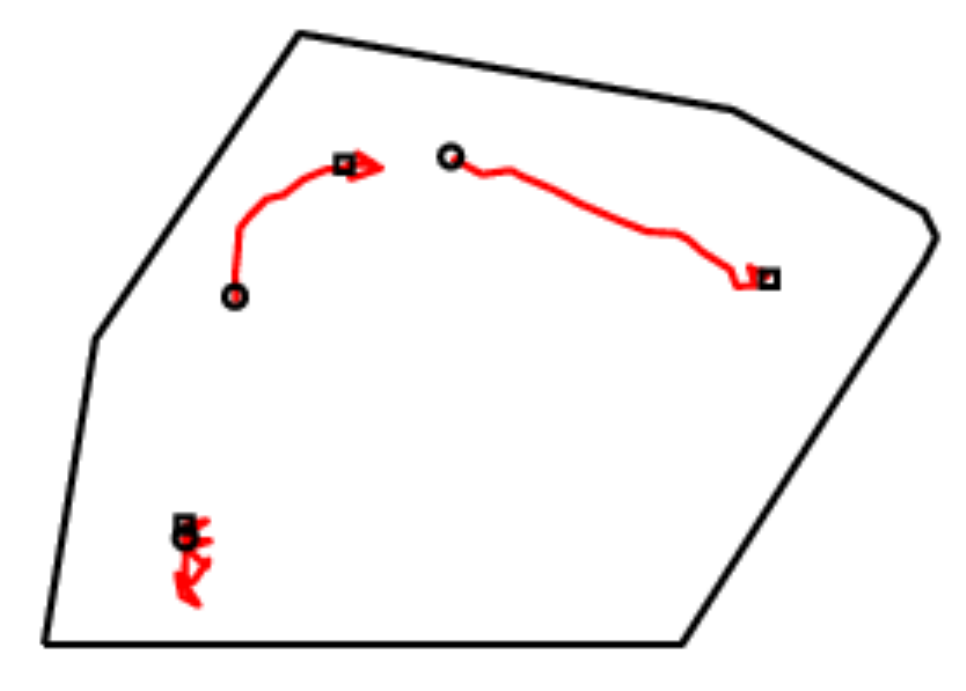}
	\else
		\includegraphics[width=0.48\textwidth]{figures/Experiment_II/exp2_traj.pdf}
		\includegraphics[width=0.48\textwidth]{figures/Experiment_II/exp2_target.pdf}
	\fi
	
	\caption{Experiment II: [Left] ArUco (green) and encoder (red) reported trajectories. [Right] Target point trajectories. Initial positions marked with circles and final positions with squares.}
	\label{fig:exp2_traj}
\end{figure}


\section{Conclusion}
This article examines the area coverage problem by a homogeneous team of mobile agents with imprecise localization. A gradient ascent based control law was designed based on a Guaranteed Voronoi partition of the region. Because of the complexity of the resulting control law, a simpler suboptimal control law is proposed and the performance of both is compared in simulation studies. Additionally, two experiments were conducted to highlight the efficiency of the suboptimal control law.


\appendices
\appendix[]
We assume $n$ disks $D_i, ~i \in I_n$ with centers $q_i = \left[ x_i, y_i \right]^T \in \mathbb{R}^2$ and radii $r_i^u, ~i \in I_n$.

\subsection{Hyperbola Definition and Properties}
\label{app:hyperbola_properties}
We will use the GV diagram definition given in (\ref{eq:GV_definition_max_min}) for two disks $D_i$ and $D_j$ and shown that the boundaries of their GV cells are hyperbola branches.

Let us find the boundary of the cell of node $i$, $V_i^g = H_{ij}$,  since we examine only two nodes. We have that
\small
\begin{equation*}
	\partial H_{ij} = \left\{ q\in\Omega\colon \max\left\|q-b_i\right\|
	= \min\left\|q-b_j\right\|, ~\forall b_i \in D_i, ~\forall b_j \in D_j
	\right\},
\end{equation*}
\normalsize
however since $D_i$ and $D_j$ are disks, the maximum and minimum distances can be calculated and thus
\small
\begin{equation*}
	\partial H_{ij} = \left\{ q\in\Omega\colon \left\|q-q_i\right\| + r_i^u
	= \left\|q-q_j\right\| - r_i^u \right\},
\end{equation*}
\normalsize
where $q_i$ and $q_j$ are the centers of $D_i$ and $D_j$ respectively. The equation
\begin{equation*}
	\partial H_{ij} = \left\{ q\in\Omega\colon \left\|q-q_j\right\| -\left\|q-q_i\right\|
	= r_i + r_j \right\},
\end{equation*}
defines one branch of a hyperbola since it is the locus of points whose difference of distances from two foci $q_i$ and $q_j$ is equal to a positive constant $r_i + r_j$. The other branch of the hyperbola is 
\begin{equation*}
	\partial H_{ji} = \left\{ q\in\Omega\colon \left\|q-q_i\right\| -\left\|q-q_j\right\|
	= r_i + r_j \right\},
\end{equation*}
and corresponds to the boundary of $V_j^g$.

In a hyperbola, the positive constant is called the major axis and is usually denoted as $2a$ whereas the distance between its foci is denoted as $2c$. A third parameter, the minor axis denoted as $2b$ is then defined by the equation $a^2+b^2=c^2$. As such in this case we have that $2a = r_i + r_j$ and $2c = \parallel q_i-q_j \parallel$. Additionally, the closest points on the two branches are called the vertices of the hyperbola and their distance is $2a$.

\begin{figure}[htb]
	\centering
	\ifx\singlecol\undefined
		\includegraphics[width=0.25\textwidth]{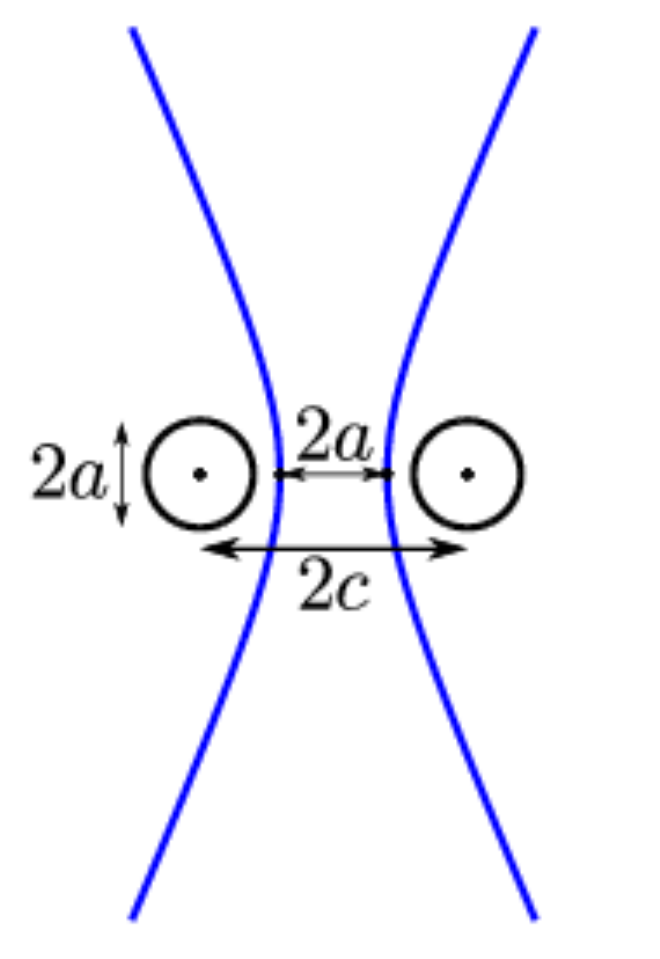}
	\else
		\includegraphics[width=0.5\textwidth]{figures/appendix/hyperbola_text.pdf}
	\fi
	
	\caption{A hyperbola with its parameters $2a$ and $2c$ shown.}
	\label{fig:hyperbola}
\end{figure}

\subsection{Hyperbola Parametric Equation}

Assuming the foci of the hyperbola are $q_i = [-c ~0]^T$  and $q_j = [c ~0]^T$, its parametric equation is
\begin{equation*}
\gamma(t) = \left[ \begin{array}{c} \pm a \cosh(t) \\ b \sinh(t) \end{array} \right], ~t \in \mathbb{R}
\end{equation*}
where $a$, $b$ and $c$ are those defined in Appendix \ref{app:hyperbola_properties} and the $+ (-)$ sign corresponds to the West (East) branch of the hyperbola which is $H_{ji} (H_{ij})$. 

By rotating the hyperbola foci around the origin by an angle $\theta$, the parametric equation becomes
\begin{equation*}
\gamma(t) = 
\left[ \begin{array}{cc}
\cos(\theta) &-\sin(\theta) \\
\sin(\theta) &\cos(\theta)
\end{array} \right] 
\left[ \begin{array}{c} \pm a \cosh(t) \\ b \sinh(t) \end{array} \right], ~t \in \mathbb{R}
\end{equation*}

Finally, given two foci $q_i = [x_i ~ y_i]^T$  and $q_j = [x_j ~ y_j]^T$ anywhere on the plane, the general hyperbola parametric equation is
\begin{equation}
\gamma(t) = 
\left[ \begin{array}{cc}
\cos(\theta) &-\sin(\theta) \\
\sin(\theta) &\cos(\theta)
\end{array} \right] 
\left[ \begin{array}{c} \pm a \cosh(t) \\ b \sinh(t) \end{array} \right]
+ \left[ \begin{array}{c}
\frac{x_i+x_j}{2} \\ \frac{y_i+y_j}{2}
\end{array} \right], ~t \in \mathbb{R}
\end{equation}
where $\theta = \arctan( \frac{y_j - y_i}{x_j - x_i} )$.

\subsection{Outward Unit Normal Vectors}

For any curve $\gamma(t)$ the normal vector is defined as
\begin{equation*}
\hat{n}(t) = \ddot{\gamma} - \left( \ddot{\gamma} \cdot \frac{\dot{\gamma}}{\parallel \dot{\gamma} \parallel}\right) \frac{\dot{\gamma}}{\parallel \dot{\gamma} \parallel}.
\end{equation*}
The magnitude of this vector is the curvature at that particular point on the curve and it points towards the center of curvature. Since the hyperbolic branches define convex regions on the plane, the outward unit normal vector is given by
\begin{equation}
n(t) = -\frac{\hat{n}}{\parallel \hat{n} \parallel}.
\end{equation}
However, since the resulting expressions are too large to show here, an indicative plot of the outward unit normal vectors of one branch of a hyperbola is shown in Figure \ref{fig:normals_hyp}.

\begin{figure}[htb]
	\centering
	\ifx\singlecol\undefined
		\includegraphics[width=0.5\textwidth]{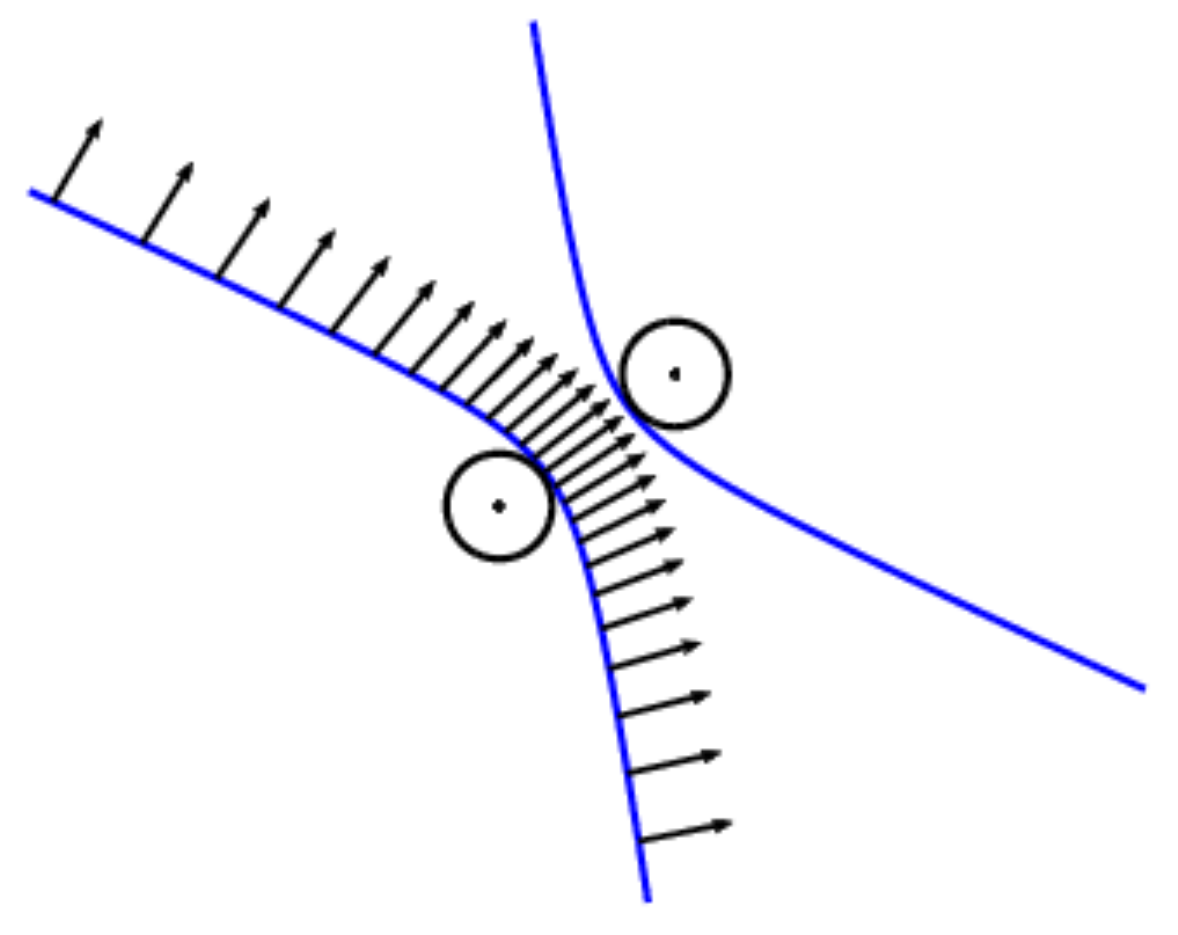}
	\else
		\includegraphics[width=0.9\textwidth]{figures/appendix/normal_vectors_hyperbola.pdf}
	\fi
	
	\caption{The outward unit normal vectors of one branch of the hyperbola are shown in black.}
	\label{fig:normals_hyp}
\end{figure}

\subsection{Jacobian Matrix}

The Jacobian matrix ${\upsilon_j^i}^T$ shows the change in the curve $\gamma_j(t) = \left[ \begin{array}{c} \gamma_{jx}(t) \\ \gamma_{jy}(t) \end{array} \right]$ caused by the movement of node $i$, thus its transpose is
\begin{equation}
\upsilon_j^i = \left[ \begin{array}{cc}
\frac{\partial \gamma_{jx}}{\partial x_i} & \frac{\partial \gamma_{jx}}{\partial y_i} \\
\frac{\partial \gamma_{jy}}{\partial x_i} & \frac{\partial \gamma_{jy}}{\partial y_i}
\end{array}\right]
\label{jacobian}
\end{equation}
The expressions for the Jacobian matrix elements are too large to show here and as such some indicative plots of them with respect to the parameter $t$ are shown in Figure \ref{fig:jacobian}. 

\begin{figure}[htb]
	\centering
	\ifx\singlecol\undefined
		\subfloat[]{ \includegraphics[width=0.23\textwidth]{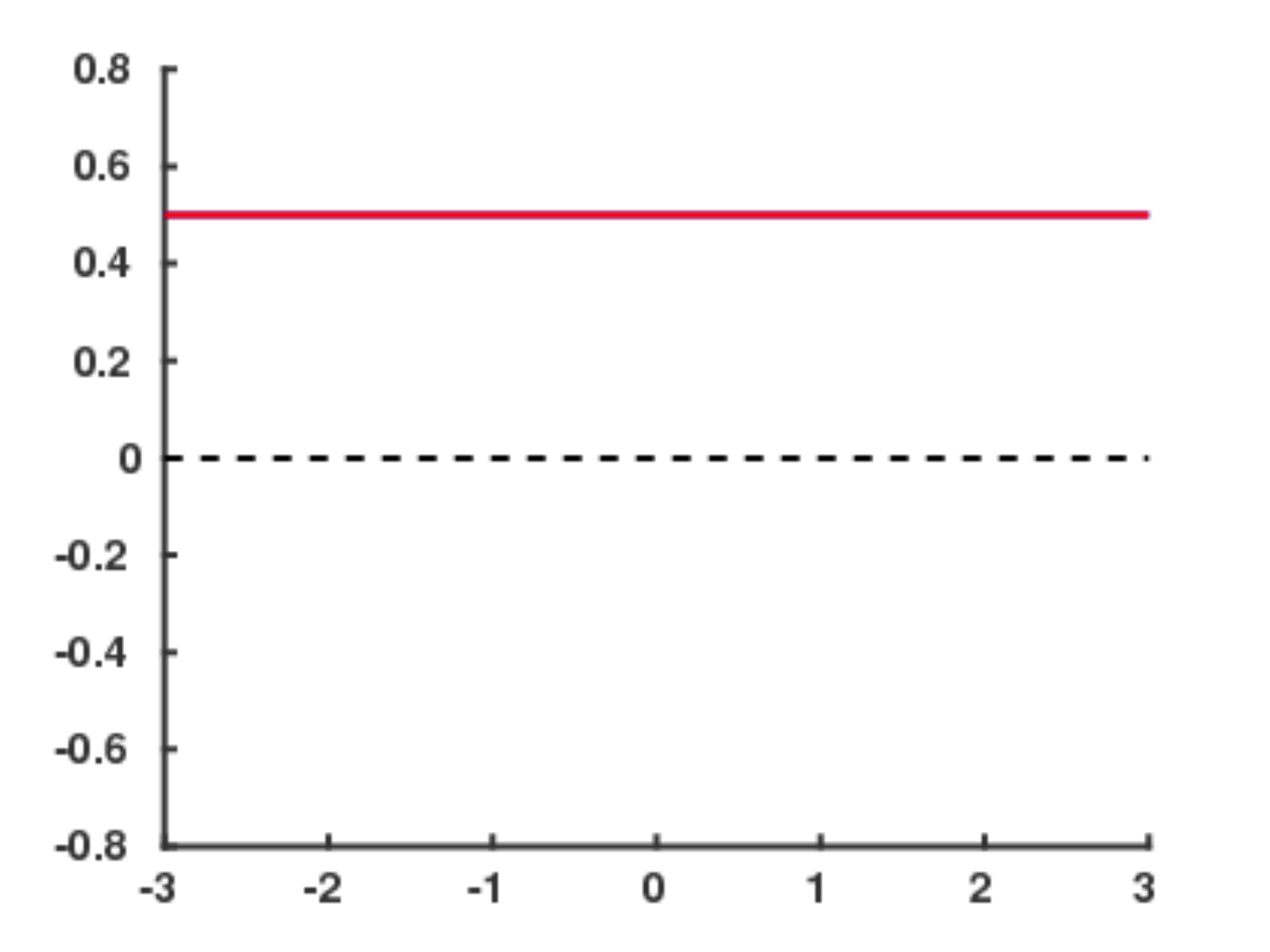} }
		\subfloat[]{ \includegraphics[width=0.23\textwidth]{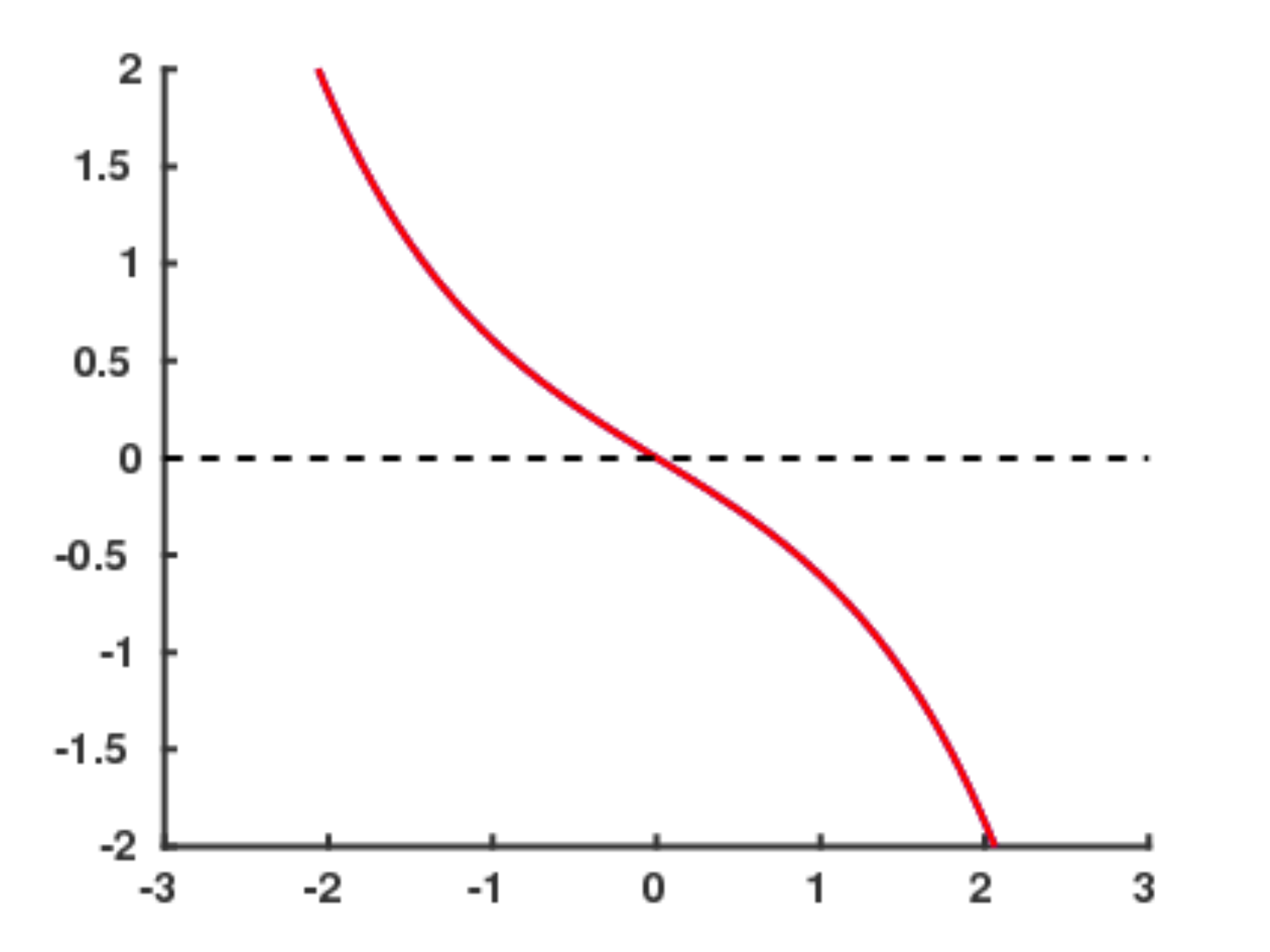} }\\
		\subfloat[]{ \includegraphics[width=0.23\textwidth]{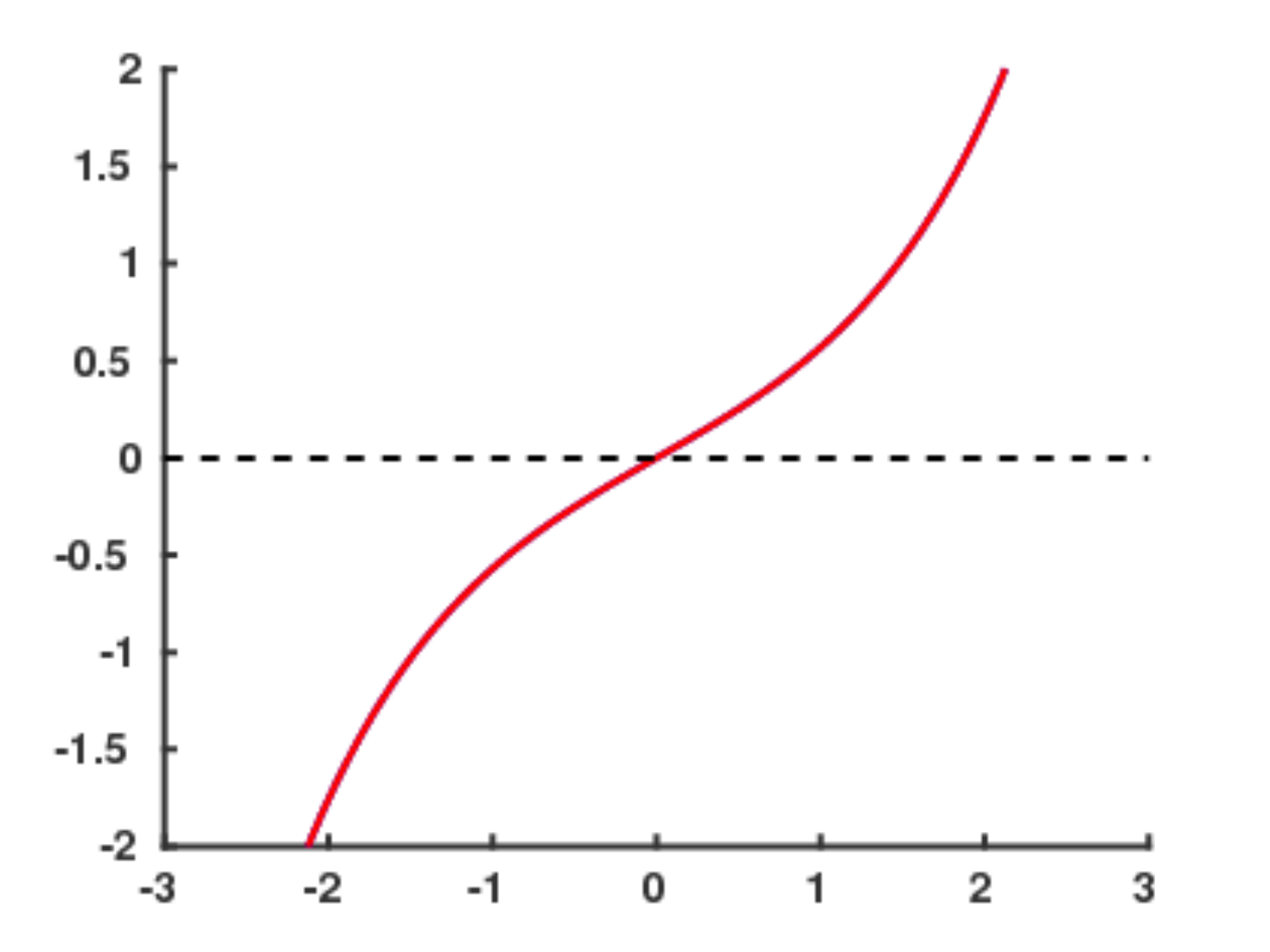} }
		\subfloat[]{ \includegraphics[width=0.23\textwidth]{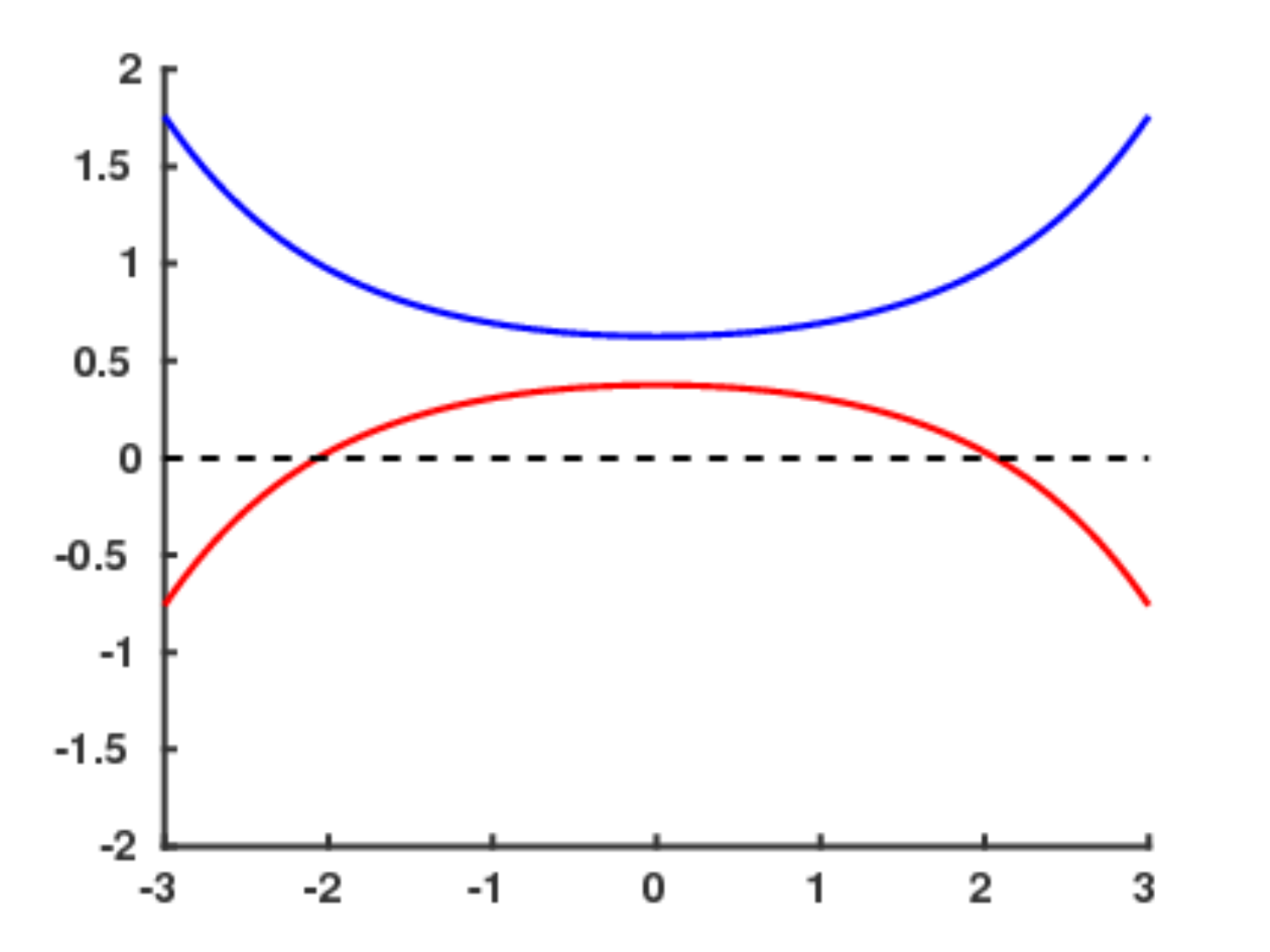} }
	\else
		\subfloat[]{ \includegraphics[width=0.45\textwidth]{figures/appendix/J_x_xi.pdf} }
		\subfloat[]{ \includegraphics[width=0.45\textwidth]{figures/appendix/J_x_yi.pdf} }\\
		\subfloat[]{ \includegraphics[width=0.45\textwidth]{figures/appendix/J_y_xi.pdf} }
		\subfloat[]{ \includegraphics[width=0.45\textwidth]{figures/appendix/J_y_yi.pdf} }
	\fi
	
	\caption{The Jacobian matrix $\upsilon_i^i$ (blue) and $\upsilon_j^i$ (red) elements as functions of $t$. (a) $\frac{\partial \gamma_{jx}}{\partial x_i}$, (b) $\frac{\partial \gamma_{jx}}{\partial y_i}$, (c) $\frac{\partial \gamma_{jy}}{\partial x_i}$, (d) $\frac{\partial \gamma_{jy}}{\partial y_i}$.}
	\label{fig:jacobian}
\end{figure}

\ifCLASSOPTIONcaptionsoff
  \newpage
\fi

\bibliographystyle{IEEEtran}
\bibliography{bibliography/SP_bibliography}

%

%
%
%
%





\end{document}